\documentclass[journal,twoside,draftclsnofoot,onecolumn]{IEEEtran}
\usepackage[mathscr]{euscript}
\usepackage[sort,compress]{cite}

\ifCLASSINFOpdf
	\usepackage[pdftex]{graphicx}
	\usepackage{color}
\else
\fi

\graphicspath{{./fig/}}

\usepackage{mathtools}
\usepackage{amssymb}

\ifCLASSOPTIONcompsoc
  \usepackage[caption=false,font=normalsize,labelfont=sf,textfont=sf]{subfig}
\else
  \usepackage[caption=false,font=footnotesize]{subfig}
\fi

\usepackage{url}
\usepackage{hyperref}

\hypersetup{
	colorlinks=true,
	linkcolor=black,      
    citecolor=black,      
    filecolor=black,      
    urlcolor=black        
}

\usepackage[disable]{todonotes}

\usepackage{amsthm}
\usepackage{nicefrac}

\newtheorem{theorem}{Theorem}
\newtheorem{props}{Proposition}
\newtheorem{lem}{Lemma}
\newtheorem{coro}{Corollary}
\newtheorem{defn}{Definition}

\newcommand{\Ker}[1]{\text{Ker}\left(#1 \right)}

\newcommand{\dd}{{\rm d}}

\newcommand{\Xone}{\mathbf{X}}
\newcommand{\Xino}{{X}}
\newcommand{\EMone}{$({\rm m}_1)$}
\newcommand{\EMtwo}{$({\rm m}_2)$}
\newcommand{\Yone}{{\mathbf{Y}}}
\newcommand{\Yino}{{Y}}
\newcommand{\Rn}{{\mathbb{R}}^n}
\newcommand{\Rm}{{\mathbb{R}}^m}
\newcommand{\y}{y}

\usepackage[yyyymmdd,hhmmss]{datetime}
\usepackage{units}
\usepackage{amsmath}

\renewenvironment{proof}[1][\proofname]{{\noindent \bfseries #1.}}{
\qed\vspace{8pt} \mbox{}
}

\newcommand{\mytitle}{Low-Complexity Multiclass Encryption by Compressed Sensing}

\begin{document}
\title{Low-Complexity Multiclass Encryption \\ by Compressed Sensing}

\author{
\IEEEauthorblockN{
Valerio Cambareri, \IEEEmembership{Student Member, IEEE}, Mauro Mangia, \IEEEmembership{Member, IEEE}, \\ Fabio Pareschi, \IEEEmembership{Member, IEEE}, Riccardo Rovatti, \IEEEmembership{Fellow, IEEE}, Gianluca Setti, \IEEEmembership{Fellow, IEEE}}

\thanks{Copyright \copyright 2015 IEEE. Personal use of this material is permitted. However, permission to use this material for any other purposes must be obtained from the IEEE by sending a request to \url{pubs-permissions@ieee.org}.}
\thanks{V. Cambareri and R. Rovatti are with the Department of Electrical, Electronic and Information Engineering (DEI), University of Bologna, Italy (e-mail: \url{valerio.cambareri@unibo.it}, \url{riccardo.rovatti@unibo.it}).}%
\thanks{M. Mangia is with the Advanced Research Center on Electronic Systems (ARCES), University of Bologna, Italy (e-mail: \url{mmangia@arces.unibo.it}).}%
\thanks{F. Pareschi and G. Setti are with the Engineering Department in Ferrara (ENDIF), University of Ferrara, Italy (e-mail: \url{fabio.pareschi@unife.it}, \url{gianluca.setti@unife.it}).}%
}

\markboth{IEEE Transactions on Signal Processing}{Cambareri \MakeLowercase{\emph{et al.}:} \textsc{\mytitle}}

\maketitle

\begin{abstract}

The idea that compressed sensing may be used to encrypt information from unauthorised receivers has already been envisioned, but never explored in depth since its security may seem compromised by the linearity of its encoding process. 

In this paper we apply {this simple} encoding to define a general {private-key} encryption scheme in which a transmitter distributes the same encoded measurements to receivers of different classes, which are provided partially corrupted encoding matrices and are thus allowed to decode the acquired signal at provably different levels of recovery quality.

The security properties of this scheme are thoroughly analysed: firstly, the properties of our multiclass encryption are theoretically investigated by deriving {performance bounds} on the recovery quality attained by lower-class receivers {with respect to high-class ones}. Then we perform a statistical analysis of the {measurements} to show that, although not perfectly secure, compressed sensing grants some level of security that comes at almost-zero cost and {thus may benefit} resource-limited applications. 

In addition to this we report some exemplary applications of multiclass encryption by compressed sensing of speech signals, electrocardiographic tracks and images, in which quality degradation is quantified as the impossibility of some feature extraction algorithms to obtain sensitive information from suitably degraded signal recoveries.
\end{abstract}


\IEEEpeerreviewmaketitle

\section{Introduction}
\IEEEPARstart{W}{ith} the rise of paradigms such as wireless sensor networks \cite{akyildiz2002survey} where a large amount of data is locally acquired by sensor nodes and transmitted {remotely} for further processing, defending the privacy of digital data gathered and distributed by {such networks} is a relevant issue. {This privacy requirement is normally met} by means of encryption stages securing the transmission channel \cite{wang2006wsn}, implemented in the digital domain and {preceded by} analog-to-digital conversion of the signal. Due to their complexity, these cryptographic modules (\emph{e.g.} those implementing {the Advanced Encryption Standard (AES)} \cite{daemen2002design}) may require a considerable amount of resources, especially in terms of power consumption.

Compressed Sensing (CS) \cite{Donoho_2006, Candes_2008} is a {mature} signal processing technique used in the development of novel data acquisition schemes. CS exploits the structure of certain signals to simultaneously perform data compression and acquisition at the physical interface between the analog and digital domain, thus allowing acquisition at sub-Nyquist rates \cite{tropp2010beyond, mishali2011xampling}. This efficient acquisition is commonly followed by a decoding algorithm that maps an undersampled set of CS-encoded measurements into a recovery of the original signal. 
Within this framework, many sensing architectures have been proposed for {the acquisition of a variety of} signals \cite{Haboba_JETCAS_2012, duarte2008single, chandrakasan2012JSSC}.

We investigate on the possibility of using CS {with Bernoulli random encoding matrices \cite{candes2006near}} as a physical-layer method to embed security properties in the acquisition process. Although it is well known that {CS cannot be regarded as perfectly secure \cite{Allerton_2008}} we will formalise its main weaknesses and strengths {as} an exploration of the trade-off between {achievable} security properties and resource requirements in low-complexity {acquisition} systems, for which an almost-zero cost encryption mechanism is an {appealing} option.

In more detail, {we here devise an encryption strategy relying} on the fact that any receiver 	{attempting to decode} the CS measurements must {know the true} encoding matrix used in the acquisition process {to attain} exact signal recovery. In {partial or complete} defect of this information, the recovered signal will be subject to a significant amount of recovery noise \cite{herman2010general}. 

We exploit this decoder-side sensitivity to provide multiple {recovery quality-based} levels {(\emph{i.e.} \emph{classes})} of {access to the information carried in the signal}. {In fact}, when {the true encoding matrix} is completely unknown the signal {is fully encrypted}, whereas if a receiver {knows it up to some random perturbations} the signal will {still} be recoverable, albeit with limited quality. 
{We therefore aim to control the recovery performances of users (receivers) belonging to the same class by exploiting their ignorance of the true encoding matrix. Since} these encoding matrices are generated from the available private keys at the corresponding decoders, high-class receivers are given a complete key and thus the true encoding matrix, lower-class receivers are given an incomplete key resulting in a partially corrupted encoding matrix. 
{To ensure that this mismatch goes undetected by lower-class receivers, we only alter the sign of a randomly chosen subset of the entries of the true encoding matrix, which is itself assumed to be a realisation of a $\pm 1$-valued Bernoulli random matrix.}

This contribution is structured as follows: in Section \ref{section1} we briefly review the theoretical framework of CS, introduce the mathematical model of two-class and multiclass CS, and {perform upper and lower bound analyses} on the recovery error norm suffered by lower-class receivers depending on the chosen amount of perturbation.

Section \ref{section2} addresses the robustness of {CS with universal random encoding matrices\cite{candes2006near}} against straightforward statistical attacks. While CS is not perfectly secret in the Shannon sense \cite{Allerton_2008} and in general suffers from continuity due to its linear nature, we 
prove that, asymptotically, nothing can be inferred about the encoded signal except for its power and formalise this fact in a relaxed secrecy condition. 
{Moreover, we show how the convergence to this behaviour is sharp for finite signal dimensions}. Hence, eavesdroppers are practically unable to extract but a very limited amount of information from the sole statistical analysis of CS measurements. Other non-statistical attacks of a more threatening nature will be treated separately in a future contribution.

In Section \ref{section4} we propose example applications of multiclass CS to concealing sensitive information in images, electrocardiographic tracks and speech signals. The recovery performances are evaluated {in a signal processing perspective} to prove the efficacy of this strategy at integrating {some security properties} in {the sensing process}, with the additional degree of freedom of allowing multiple quality levels and eventually hiding selected signal features to certain user classes.

\subsection{Relation to Prior Work}
This contribution mainly improves on two separate lines of research: ($i$) statistical security analyses of the CS encoding and ($ii$) the effect of encoding matrix perturbations on signal recovery. Line ($i$) stems from the {security analysis} in \cite{Allerton_2008}. Both \cite{Allerton_2008} and \cite{orsdemir2008security} showed how brute-force attacks are computationally infeasible, so that some security properties could indeed be provided in relevant applications \cite{cambareri2013twoclass, zhang2013VLSI}. 
We deepen the results in \cite{Allerton_2008} by introducing an asymptotic notion of secrecy for signals having the same power, and by verifying it for CS. We assess its consequences for finite dimensions by means of hypothesis testing and develop a non-asymptotic analysis of the rate at which acquired signals having the same energy become indistinguishable when observing the probability distribution of their subgaussian CS measurements. 
This is obtained by adapting a recent result in probability theory \cite{klartag2012variations}. 
On the other hand, line ($ii$) relates to studying the effect of our particular sparse random perturbation matrix on signal recovery performances, a field that is closely related to statistical sparse regression with corrupted or missing predictors \cite{loh2012corrupted, loh2012high}. The authors of \cite{herman2010general} quantify this effect in a general framework: 
we will adapt these results to our case for a formal analysis of the \emph{worst-case} signal recovery performances of lower-class users, while developing some new arguments to find their \emph{best-case} recovery performances -- our aim being the distinction between different user classes. 

\section{Multiclass Compressed Sensing}
\label{section1}
\subsection{Brief Review of Compressed Sensing}
\label{prelim}

Compressed sensing \cite{Donoho_2006, Candes_2008} {is summarised by} considering the following setting: let $x$ be a {vector} in $\Rn$ and $y \in \Rm$ a vector of \emph{measurements} obtained from $x$ by applying a linear dimensionality-reducing transformation $y = A x$, \emph{i.e.} 
\begin{equation}
y_j = \sum^{n-1}_{l = 0} A_{j,l} x_l \, , \ j = {0, \ldots, m-1}
\label{measures}
\end{equation}
\noindent with $A_{m \times n}$ the \emph{encoding matrix}. 
Under suitable assumptions, fundamental results \cite{Donoho_2006, Candes_2005, Candes_2006} showed it is possible to recover $x$ from $y$ even if $m < n$. 

The first of such assumptions is that $x$ has a {$k$-\emph{sparse} representation}, \emph{i.e.} we assume that {there exists} a \emph{sparsity basis} $D_{n \times n}$ {such that $x = D s$}, {with $s \in \Rn$ having a support of cardinality $k$}. This cardinality is also indicated as $\|s\|_0 = k$ with $k \ll n$. 
Asserting that $x$ is represented by $k < m < n$ non-zero coefficients in a suitable domain intuitively means that its intrinsic information content is smaller than the apparent dimensionality. In the following we will assume that $D$ is an orthonormal basis (ONB).

A second assumption must be made on the structure of $A$. {Many conditions have been formulated in the literature (\emph{e.g.} the \emph{restricted isometry property}, RIP \cite{candes2008restricted}) to guarantee that the information in $s$ is preserved through the mapping $y =A D s$.} 
To the purpose of this paper it suffices to say that the most \emph{universal} option (\emph{i.e.} independently of $D$) is choosing $A$ as typical realisations of a random matrix with {independent and identically distributed} (i.i.d.) entries from a subgaussian distribution, \emph{e.g.} an i.i.d. Gaussian or Bernoulli random matrix \cite{candes2006near}. 
We will let $A$ be an ${m\times n}$ i.i.d. Bernoulli random matrix\footnote{This notation is used both for the random matrix and its realisations, disambiguating with the term \emph{instance} where needed.} unless otherwise noted. 

When both these conditions hold, $s$ can be recovered from $y = A D s$ as the sparsest vector solving the hard problem
\begin{equation}
{s} = \mathop{\text{arg} \, \min}_{\xi \in \mathbb{R}^n} \|\xi\|_0 \ \text{s. t.} \ y = A D \xi \tag{$P_0$}
\label{eq:minl0}
\end{equation}
Moreover, if the dimensionality of the measurements $m \sim k \log \frac{n}{k}$ is not too small w.r.t. that of $x$ and its sparsity w.r.t. $D$, \eqref{eq:minl0} can be relaxed to the convex $\ell_1$-norm minimisation
\begin{equation}
\hat{s} = \mathop{\text{arg} \, \min}_{\xi \in \mathbb{R}^n} \|\xi\|_1 \ \text{s. t.} \ y = A D \xi \tag{$P_1$}
\label{eq:minl1}
\end{equation}
\noindent still yielding $\hat{s} = s$ (provided that $A$ is carefully chosen \cite{Candes_2006}), with \eqref{eq:minl1} being a linear programming problem {solved with} polynomial-time algorithms \cite{Candes_2005}. In the following, we will refer to this problem as the $\min \, \ell_1$ {decoder} (also known as \emph{basis pursuit}, BP).

\subsection{A Cryptographic Perspective} 
\label{cryptopersp}
Standard CS {may be interpreted} as a private key cryptosystem where $x$ is the \emph{plaintext}, the measurement vector $y$ is the \emph{ciphertext} and the \emph{encryption algorithm} is a linear transformation operated by the encoding matrix $A$ defining the acquisition process. In the classic setting, Alice acquires a plaintext $x$ by CS using $A$ and sends to Bob the ciphertext $y$; Bob is able to successfully recover $x$ from $y$ if he is provided with $A$ or equivalently the private key required to generate it.

Since {many} CS-based acquisition systems \cite{Haboba_JETCAS_2012, duarte2008single} entail the use of i.i.d. Bernoulli random matrices generated by a pseudorandom number generator (PRNG) we define \emph{encryption key} (or \emph{shared secret}) the initial seed which is expanded by the PRNG to generate a sequence of encoding matrices. In the following we will assume that the period of this sequence is sufficiently long to guarantee that in a reasonable observation time no two encoding matrices will be the same, \emph{i.e.} that each plaintext will be encoded with a different matrix. With this hypothesis, we let each instance of $A$ be a generic, unique element of the aforementioned sequence.

\subsection{Signal Models and Assumptions}
This paper will analyse the security properties of CS starting from some statistical properties of the signal being encoded as in \eqref{measures}. 
Rather than relying on its \emph{a priori} distribution, our analysis uses general moment assumptions that may correspond to many probability distributions on the signal domain. We will therefore adopt the following signal models:
\begin{enumerate}
\item[\EMone] for finite $n$, we let $\Xino = \{X_j\}^{n-1}_{j=0}$ be a real random vector (RV). Its realisations (finite-length plaintexts) $x = \begin{pmatrix} x_0, \cdots, x_{n-1} \end{pmatrix} \in \Rn$ are assumed to have finite energy $e_x = \|x\|^2_2$. We will let each $x = D s$ with $D$ an ONB and $s$ being $k$-sparse to comply with sparse signal recovery guarantees \cite{Candes_2006}. 
$X$ is mapped to the measurements' RV $Y =\{Y_j\}^{m-1}_{j = 0}$ (whose realisations are the ciphertexts $y$) as $Y = A X$, \emph{i.e.} each realisation of $(Y, A, X)$ is an instance of \eqref{measures}. 

\item[\EMtwo] for $n \rightarrow \infty$, we let $\Xone = \{X_j\}^{+\infty}_{j=0}$ be a real random process (RP). Its realisations (infinite-length plaintexts) $x$ are assumed to have finite power $W_x = \lim_{n \rightarrow \infty} \frac{1}{n} \sum^{n-1}_{j = 0} x^2_j$. We denote them as sequences ${x} = \{x^{(n)}\}^{+\infty}_{n = 0}$ of finite-length plaintexts $x^{(n)} = \begin{pmatrix} x_0, \cdots, x_{n-1} \end{pmatrix}$. $\Xone$ is mapped to either a RV $\Yino$ of realisations (ciphertexts) $y$ for finite $m$, or a RP $\Yone = \{Y_j\}^{+\infty}_{j=0}$ of ciphertexts $y$ for $m,n  \rightarrow \infty, \frac{m}{n} \rightarrow q$. Both cases are comprised of random variables $Y_j = \frac{1}{\sqrt{n}} \sum^{n-1}_{l=0} A_{j,l} X_l$.  {The $\frac{1}{\sqrt{n}}$ scaling is not only theoretically needed for normalisation purposes, but also practically required in the design of finite quantiser ranges for CS-based acquisition front-ends.}
\end{enumerate}
When none of the above models is specified, a single realisation of \eqref{measures} is considered as in the standard CS framework (Section \ref{prelim}).

\subsection{Multiclass Encryption by Compressed Sensing}
\label{multisec}

Let us consider a scenario where multiple users receive the same measurements $y$, know the sparsity basis $D$, but are made different by the fact that some of them know the true $A$, whereas the others only know an approximate version of it. The resulting mismatch between $A$ and its approximation used in the decoding process by the latter set of receivers will limit the quality of signal recovery as detailed below.

\subsubsection{Two-Class Scheme}\label{twoclssch} With this principle in mind a straightforward method to introduce perturbations is flipping the sign of a subset of the entries of the encoding matrix in a random pattern. More formally, let $A^{(0)}$ denote the \emph{initial encoding matrix} and $C^{(0)}$ a subset of $c < m \cdot n$ {index pairs} chosen at random for each $A^{(0)}$. We construct the \emph{true encoding matrix} $A^{(1)}$ by 
\[
A^{(1)}_{j,l} = \begin{cases}
A^{(0)}_{j,l}, & (j, l) \notin C^{(0)}\\
-A^{(0)}_{j,l}, & (j, l) \in C^{(0)}
\end{cases}
\]
\noindent and use it to encode $x$ as in \eqref{measures}. Although this alteration simply involves inverting $c$ randomly chosen sign bits in a buffer of $m \cdot n$ pseudorandom symbols, we will use its linear model
\begin{equation}
\label{eq:pert}
A^{(1)} = A^{(0)} + \Delta A
\end{equation}
\noindent where $\Delta A$ is a $c$-sparse random perturbation matrix of entries
\begin{equation}
\label{eq:spperturb}
\Delta A_{j, l} = \begin{cases}
0, & (j, l) \notin C^{(0)}\\
-2 A^{(0)}_{j, l}, & (j, l) \in C^{(0)}
\end{cases}
\end{equation}
or equivalently
\begin{equation}
\label{eq:spperturb1}
\Delta A_{j, l} = \begin{cases}
0, & (j, l) \notin C^{(0)}\\
2 A^{(1)}_{j, l}, & (j, l) \in C^{(0)}
\end{cases}
\end{equation}
\noindent with density $\eta = \frac{c}{m n}$, \emph{i.e.} the ratio of non-zero entries w.r.t. the product of the dimensions of $\Delta A$. By doing so, any receiver is still provided an encoding matrix differing from the true one by an instance of $\Delta A$. This perturbation is \emph{undetectable}, \emph{i.e.} $A^{(1)}$ and $A^{(0)}$  are statistically indistinguishable since they are equal-probability realisations of the same i.i.d. Bernoulli random matrix ensemble \cite{candes2006near} with all points in $\{-1, 1\}^{m \times n}$ having the same probability.

{A \emph{first-class} user receiving $y = A^{(1)} x = (A^{(0)} + \Delta A) x$ and knowing $A^{(1)}$ is able to recover}, in absence of other noise sources and with $m$ sufficiently larger than $k$, the exact sparse solution $\hat{s} = s$ by solving \eqref{eq:minl1} \cite{Candes_2005, Candes_2006}. A \emph{second-class} user only knowing $y$ and $A^{(0)}$ is instead subject to an equivalent signal- and perturbation-dependent noise term $\varepsilon$ due to missing pieces of information on $A^{(1)}$, \emph{i.e.} 
\begin{equation}
\label{eq:secondclass}
y = A^{(1)} x = A^{(0)} x + \varepsilon
\end{equation}
\noindent where $\varepsilon = \Delta A x$ is a pure disturbance since both $\Delta A$ and $x$ are unknown to the second-class decoder.

In general, performing signal recovery in the erroneous assumption that $y = A^{(0)} x$, \emph{i.e.} with a corrupted encoding matrix will lead to a noisy recovery of $x$. Nevertheless, upper bounds on the recovery error norm $\|\hat{x} - x\|_2$ (with $\hat{x} = D \hat{s}$, $\hat{s}$ an approximation of $s$) are well known for measurements affected by generic additive noise \cite[Theorem 1.1, 1.2]{Candes_2006}. These bounds have been extended in \cite{herman2010general} to a general perturbed encoding matrix model that encompasses \eqref{eq:secondclass}. We adapt these results in Section \ref{ubsection} to obtain a worst-case analysis of the second class recovery error norm. 

Moreover, to prove the difference between first- and second-class recovery performances, in Section \ref{lbsection} we develop a {lower bound}, \emph{i.e.} a best-case analysis of the second-class recovery error norm. Both performance bounds show a clear dependence on the perturbation density $\eta$, which is suitably chosen to fix the desired quality range for each class.

\subsubsection{Multiclass Scheme} The two-class scheme may be iterated to devise an arbitrary number $w$ of user classes: sign-flipping is now applied on disjoint subsets of {index pairs} $C^{(u)}, u = 0, \ldots, w-2$ of $A^{(0)}$ so that
$$A^{(u+1)}_{j,l} = \begin{cases} A^{(u)}_{j,l}, & (j, l) \notin C^{(u)} \\ -A^{(u)}_{j,l}, & (j, l) \in C^{(u)}\end{cases}$$
If the plaintext $x$ is encoded with $A^{(w-1)}$ then we may distinguish \emph{high-class} users knowing the complete encoding $A^{(w-1)}$, \emph{low-class} users knowing only $A^{(0)}$ and \emph{mid-class} users knowing $A^{(u+1)}$ with $u = 0, \ldots, w-3$. This simple technique can be applied to provide multiple classes of access to the information in $x$ by having different signal recovery performances at the decoder.

\subsubsection{A System Perspective}
\label{systempersp}
\begin{figure}[t]
	\centering 
	\includegraphics[width = 2.6in]{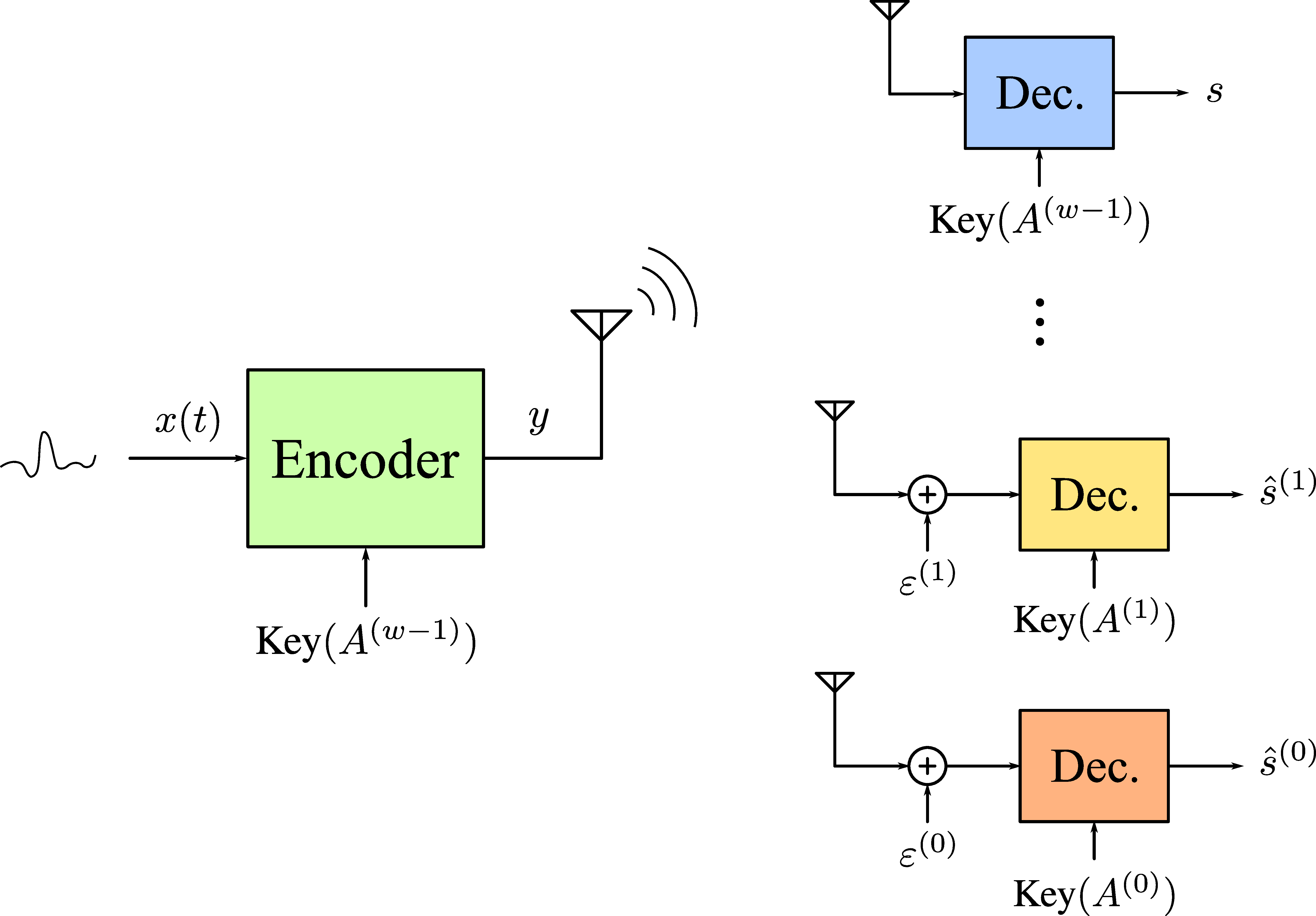}
	\caption{A multiclass CS network: the encoder acquires an analog signal $x(t)$ at sub-Nyquist rate and transmits the measurement vector $y$. Low-quality decoders reconstruct a signal approximation with partial knowledge of the encoding, resulting in additive perturbation noise $\varepsilon^{(u)}$ and leading to an approximate solution $\hat{s}^{(u)}$ for the $u$-th user class.}
	\label{fig:network}
\end{figure}

The strategy described in this section provides a multiclass encryption architecture where the shared secret between the CS encoder and each receiver is distributed depending on the quality level granted to the latter.

In particular, the full encryption key of a $w$-class CS system is composed of $w$ seeds, \emph{i.e.} low-class users are provided the secret $\text{Key}(A^{(0)})$, class-$1$ users are provided $\text{Key}(A^{(1)}) = \left(\text{Key}(C^{(0)}), \text{Key}(A^{(0)})\right)$ up to high-class users with $\text{Key}(A^{(w-1)}) = \left(\text{Key}(C^{(w-2)}), \ \cdots \ , \text{Key}(C^{(0)}), \text{Key}(A^{(0)})\right)$. A sample network scenario is depicted in Fig. \ref{fig:network}. 

From the resources point of view, multiclass CS can be enabled with very small computational overhead. The encoding matrix generator is the same at both the encoder and high-class decoder side, whereas lower-class decoders may use the same generation scheme but are unable to rebuild the true encoding due to the missing private keys $\text{Key}(C^{(u)})$. 

The initial matrix $A^{(0)}$ is updated from a pseudorandom binary stream generated by expanding  $\text{Key}(A^{(0)})$ with a PRNG. The introduction of sign-flipping is a simple post-processing step carried out on the stream buffer by reusing the same PRNG architecture and expanding the corresponding $\text{Key}(C^{(u)})$, thus having minimal computational cost.

Since the values generated by this PRNG are never exposed, cryptographically secure generators may be avoided, provided that the period with which the matrices are reused is kept sufficiently large -- this requirement is crucial to avoid attacks that could exploit multiple plaintext-ciphertext pairs to fully or partially recover the encoding matrix.

\subsection{Lower-Class Recovery Performance Bounds}
In order to quantify the recovery quality performance gap between low- and high-class users receiving the same CS measurements from the network of Fig. \ref{fig:network}, we now provide performance bounds on the recovery error in the simple two-class case, starting from the basic intuition that if the sparsity basis of $x$ is not the canonical basis, then most plaintexts $x \notin \Ker{\Delta A}$, so the perturbation noise $\varepsilon =\Delta A x \neq 0$. 

\subsubsection{Second-Class Recovery Error -- Lower Bound}
\label{lbsection}
The following results aim at predicting the best-case recovery quality of any second-class decoder that assumes $y$ was encoded by $A^{(0)}$, whereas $y = A^{(1)} x$ in absence of other noise sources and regardless of the sparsity of $x$.

\begin{theorem}[Second-class recovery error lower bound]
\label{th1}
Let:
\begin{enumerate}
\item $A^{(0)}, A^{(1)}$ be $m\times n$  i.i.d. Bernoulli random matrices as in \eqref{eq:pert} and $\Delta A$ the sparse random perturbation matrix in \eqref{eq:spperturb} of density $\eta \leq \frac{1}{2}$;

\item $X$ be as in \EMone\, with finite $\mathscr{E}_x = \mathbb{E}[\sum^{n-1}_{j=0} X^2_j]$, $\mathscr{F}_x = \mathbb{E}[(\sum^{n-1}_{j=0} X^2_j)^2]$ and $Y = A^{(1)} X$ be the corresponding measurements' RV;
\end{enumerate}
For all $\theta \in (0, 1)$ and any instance $y = A^{(1)} x$, any $\hat{x}$ that satisfies $y = A^{(0)} \hat{x}$ is such that the recovery error norm 
\begin{equation}
\label{finitererrlb}
\|\hat{x} - x\|^2_2 \geq \dfrac{4 \eta m \mathscr{E}_x}{\sigma_ {\max}(A^{(0)})^2} \, \theta
\end{equation}
with probability
\begin{equation}
\textstyle \zeta=\frac{1}{1+(1-\theta)^{-2}\left\{\left[1+\frac{1}{m}\left(\frac{3}{2\eta}-1\right)\right]\frac{\mathscr{F}_x}{\mathscr{E}_x^2}-1\right\}}
\end{equation}
where $\sigma_{\max}(\cdot)$ denotes the maximum singular value of its argument.
\end{theorem}
\begin{coro}[Asymptotic case of Theorem 1]
Let:
\begin{enumerate}
\item $A^{(0)}, A^{(1)}, \Delta A, \eta$ be as in Theorem 1 as $m, n \rightarrow \infty, \frac{m}{n} \rightarrow q$;
\item $\Xone$ be as in \EMtwo\,, {$\alpha$-mixing \cite[(27.25)]{billingsley2008probability}}, with finite $\mathscr{W}_x = \lim_{n\rightarrow\infty} \frac{1}{n} \mathbb{E}[\sum^{n-1}_{j=0} X^2_j]$ and uniformly bounded $\mathbb{E}[X^4_j] \leq m_x$ for some $m_x > 0$. Denote with $\Yone$ the corresponding measurements' RP of instances $y$;
\end{enumerate}
For all $\theta \in (0, 1)$ and $\y = \frac{1}{\sqrt{n}} A^{(1)} x$, any $\hat{x}$ that satisfies $y = \frac{1}{\sqrt{n}} A^{(0)} \hat{x}$ is such that the recovery error power
 \begin{equation}
\label{rerrlb}
W_{\hat{x}-x} = \lim_{n \rightarrow \infty} \frac{1}{n} \sum^{n-1}_{j=0}(\hat{x}_j - x_j)^2 \geq \dfrac{4 \eta q \mathscr{W}_x}{(1 + \sqrt{{q}})^2} \, \theta \end{equation}
with probability $1$.
\end{coro}

The proof of these statements is given in Appendix \ref{app1}. {Simply put, Theorem 1 and Corollary 1 state that a second-class decoder recovering 
$\hat{x}$ such that $y = A^{(0)} \hat{x}$ is subject to a recovery error whose norm, with high probability, exceeds a quantity depending on the density $\eta$ of the perturbation $\Delta A$, the undersampling rate $\frac{m}{n}$ and the average energy $\mathscr{E}_x$ or power $\mathscr{W}_x$ respectively. 
In particular, the non-asymptotic case in \eqref{finitererrlb} is a probabilistic lower bound:
as a quantitative example, by assuming it holds with probability $\zeta = 0.98$ and that $\frac{\mathscr{F}_x}{\mathscr{E}^2_x} = 1.0001, n = 1024, m = 512, \sigma_{\max}(A^{(0)}) \approx \sqrt{m} + \sqrt{n}$ (see \cite{rudelson2010non}) one could take an arbitrary $\theta = 0.1 \Rightarrow \eta = 0.1594$ to obtain $\|\hat{x}-x\|^2_2 \geq 0.0109$ w.r.t. RVs having average energy $\mathscr{E}_x = 1$. In other words, with probability $0.98$ a perturbation of density $\eta = 0.1594$ will cause a minimum recovery error norm of $\unit[19.61]{dB}$.}

A stronger asymptotic result holding with probability $1$ on $W_{\hat{x}-x}$ is then reported in Corollary 1 under mild assumptions on the RP $\Xone$, where $\theta$ can be arbitrarily close to $1$ and only affecting the convergence rate to this lower bound.
The bounds in \eqref{finitererrlb} and \eqref{rerrlb} are adopted as reference best-cases in absence of other noise sources for the second-class decoder, which exhibits higher recovery error for most problem instances and reconstruction algorithms as detailed in Section \ref{section4}. 

\subsubsection{Second-Class Recovery Error -- Upper Bound} 
\label{ubsection} 
We now derive a second-class recovery error upper bound by applying the theory in \cite{herman2010general} which extends the well-known recovery guarantees in \cite{Candes_2006} to a perturbed encoding matrix model identical to \eqref{eq:pert}. While adaptations exist \cite{yang2012robustly} none tackle the unstructured, i.i.d. sparse random perturbation of \eqref{eq:spperturb} which we will now model.

The framework of \cite{herman2010general} analyses the recovery error upper bound when using $A^{(0)} = A^{(1)}-\Delta A$ as the encoding matrix in the reference \emph{basis pursuit with denoising} (BPDN) problem
\begin{equation}
\hat{s} = \mathop{\text{arg} \, \min}_{\xi \in \mathbb{R}^n} \|\xi\|_1 \ \text{s. t.} \ \|y - A^{(0)} D \xi \|_2 \leq \gamma \tag{$P_2$}
\label{eq:bpdn}
\end{equation}
for a given noise parameter $\gamma \geq \|\varepsilon\|_2$. 
Let $\sigma^{(k)}_{\min/\max}(\cdot)$ denote the extreme singular values among all $k$-column submatrices of a matrix, and define the perturbation-related constants
\begin{equation}
\epsilon^{(k)}_{A^{(1)}} \geq \textstyle\frac{\sigma^{(k)}_{\max}(-\Delta A D)}{\sigma^{(k)}_{\max}(A^{(1)} D)}; \
\epsilon_{A^{(1)}} \geq \frac{\sigma_{\max}(-\Delta A D)}{\sigma_{\max}(A^{(1)} D)} \geq  \epsilon^{(k)}_{A^{(1)}}
\label{eq:strohmer}
\end{equation}
and the well-known \emph{restricted isometry constant} (RIC, \cite{candes2008restricted}) $\delta^{(k)} = \max{\{\sigma^{(k)}_{\max}(A^{(1)} D)^2 - 1, 1 - \sigma^{(k)}_{\min}(A^{(1)} D)^2\}}$. We estimate and plug these quantities in the upper bound of \cite[Theorem 2]{herman2010general}. For the sake of simplicity, we report it below in the case of plaintexts $x$ having exactly $k$-sparse representations in absence of other noise sources.
\begin{props}[Second-class recovery error upper bound, adapted from \cite{herman2010general}]
\label{ubprops}
Let 
\begin{enumerate}
\item $A^{(0)}, A^{(1)}, \Delta A, \eta$ be as in Theorem \ref{th1};
\item $\Xino$ be as in \EMone, with $x = D s$, $D$ an ONB and $s$ being $k$-sparse;
\item $\epsilon^{(2 k)}_{A^{(1)}} < 2^{\frac{1}{4}} - 1$ and $\delta^{(2 k)} < \delta^{(2 k)}_{\max} = \sqrt{2} (1 + \epsilon^{(2 k)}_{A^{(1)}})^{-2} - 1$;
\end{enumerate}
For any instance $y = A^{(1)} x$, a vector $\hat{x} = D \hat{s}$ with $\hat{s}$ the solution of \eqref{eq:bpdn} with noise parameter $\gamma = \epsilon^{(k)}_{A^{(1)}} \sqrt{\textstyle\frac{1+\delta^{(k)}}{1-\delta^{(k)}}}\|y\|_2$ obeys \cite{herman2010general}
\begin{equation}
\|\hat{x}-x\|_2 \leq \textstyle{\overline{C} \gamma}, \textstyle{\overline{C} = \frac{4 \sqrt{1 + \delta^{(2 k)}} (1 + \epsilon^{(2 k)}_{A^{(1)}})}{1 - (\sqrt{2} + 1) \left[(1 + \delta^{(2 k)}) (1 + \epsilon^{(2 k)}_{A^{(1)}})^2 - 1\right]}} 
\label{eq:ubrip}
\end{equation}
\end{props}
Such a guarantee depends on $\epsilon^{(k)}_{A^{(1)}}, \epsilon^{(2 k)}_{A^{(1)}}$: theoretical results exist for estimating their value by bounding the maximum singular values in \eqref{eq:strohmer} since the entries of $A^{(1)}$ and $\Delta A$ are i.i.d. (in particular, \cite{rudelson2010non} applies to $A^{(1)}$, \cite{latala2005some} to $\Delta A$). Yet, they would hold only when $D$ is the identity and involve universal constants whose values would require numerical evaluation. For these reasons we choose to estimate the required quantities directly by Monte Carlo simulation.
\begin{figure}[t]	
\centering
	\subfloat[\label{fig:evalues}Empirical values of $\epsilon^{(k)}_{A^{(1)}}$]{\includegraphics[clip=true, trim = 0 0 0 5pt, width = 1.6in]{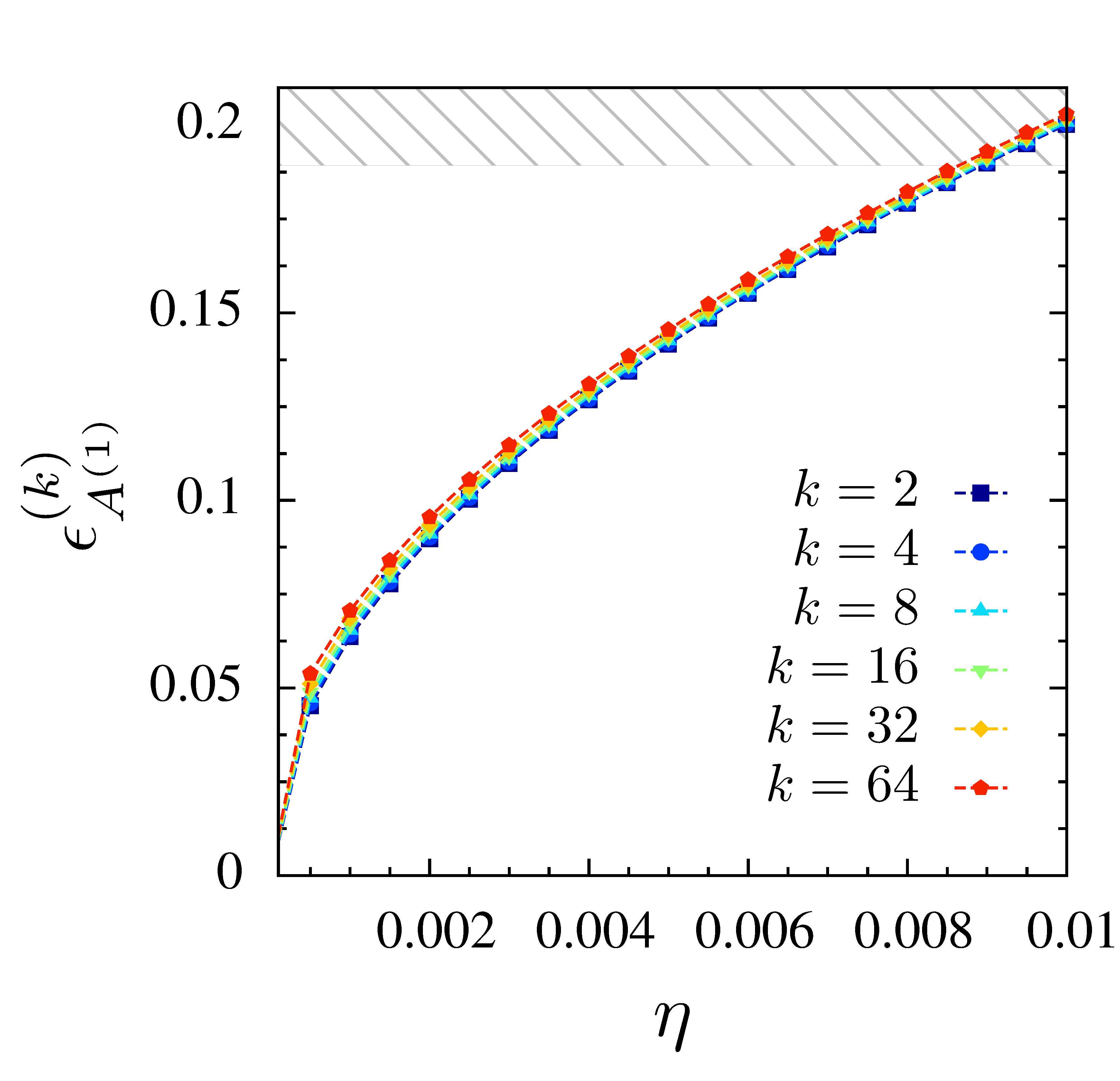}}
	\hfill
	\subfloat[\label{fig:ripc}Maximum allowed values of $\delta^{(2 k)}$]{\includegraphics[clip=true, trim = 0 0 0 5pt, width = 1.6in]{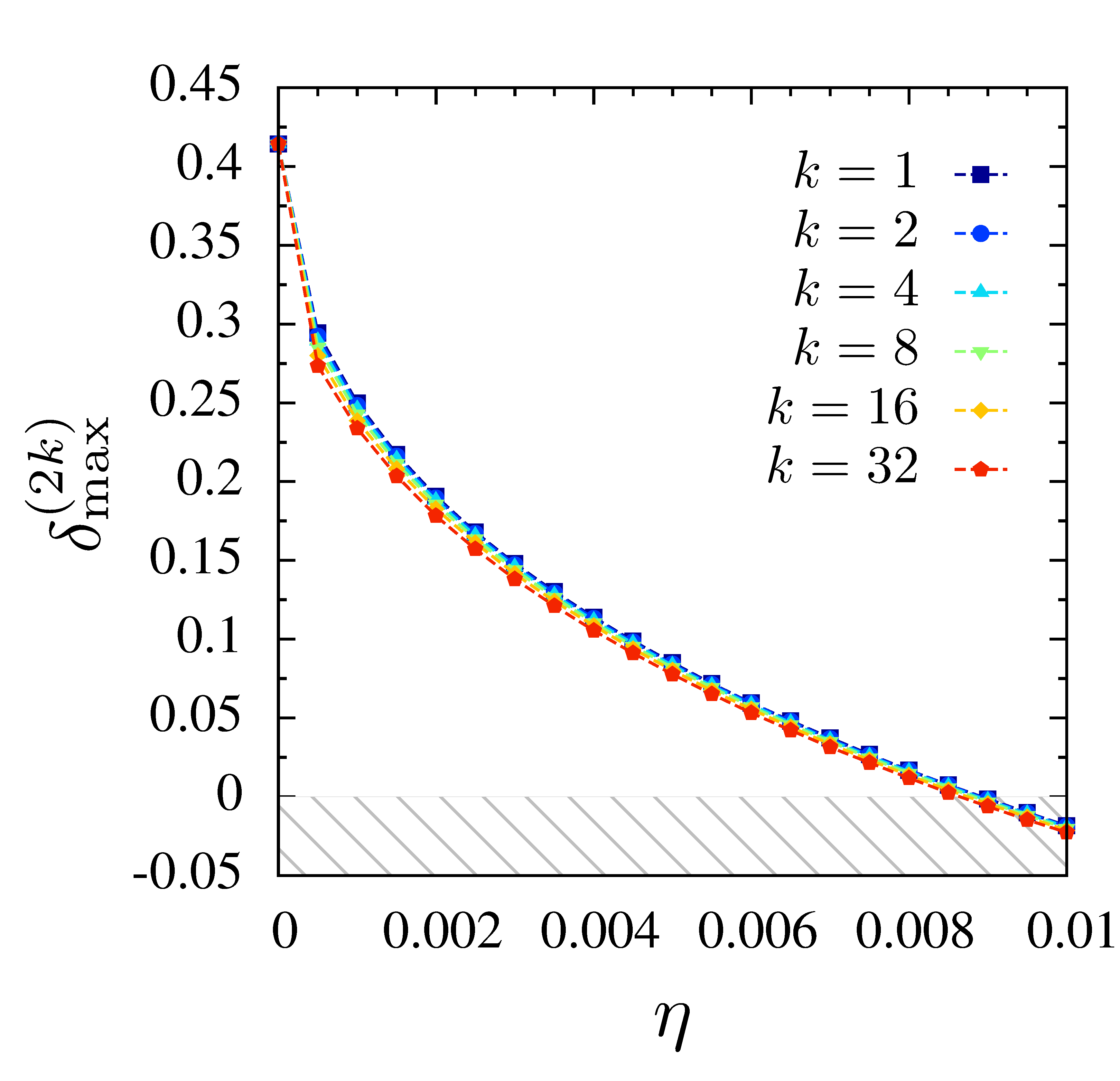}}
	\caption{Empirical evaluation of the constants in Proposition \ref{ubprops} based on a large number of $A^{(1)}, \Delta A$ with $m = 512, \eta \in [5\cdot 10^{-4},10^{-2}]$ and $D$ a random ONB.}
	\label{fig:strohmerfwk}
\end{figure}
As an example, we calculate \eqref{eq:strohmer} for $10^4$ instances of submatrices of $A^{(1)}$ and $\Delta A$ with $m = 512, k = 2, 4, \ldots, 64$ and $\eta \in [5\cdot 10^{-4},10^{-2}]$. This allows us to find typical values of ${\epsilon}^{(k)}_{A^{(1)}}$ as reported in Fig. \ref{fig:evalues}. 
In the same setting ${\epsilon}^{(k)}_{A^{(1)}} < 2^{\frac{1}{4}} - 1$ only when $\eta \leq 8 \cdot 10^{-3}$. 
In Fig. \ref{fig:ripc} we report the corresponding range of allowed RIC $\delta^{(2 k)} \leq \delta^{(2k)}_{\max}$ that comply with Proposition \ref{ubprops}, \emph{i.e.} the RIC constraints the encoding matrices must meet so that \eqref{eq:ubrip} holds. 

Such RIP-based analyses provide very strong sufficient conditions for signal recovery (see \cite[Section X]{Donoho_2010}) which in our case result in establishing a formal upper bound for a small range of $\eta$ and when solving \eqref{eq:bpdn}. {As observed by the very authors of \cite{herman2010general}, typical recovery errors are substantially smaller than this upper bound. We will therefore rely on another less rigorous, yet practically effective least-squares approach using the same hypotheses of Theorem 1 to bound the average recovery quality performances in Section \ref{section4}.}
\section{A Cryptanalysis of Compressed Sensing}
\label{section2}
Consider a generic CS scheme as in Section \ref{prelim} with $y = A x$ linearly encoding a plaintext $x$ into a ciphertext $y$. We now investigate the security properties and limits of such random linear measurements by letting $x, y$ be realisations of either RVs \EMone\, or RPs \EMtwo\, with their respective \emph{a priori} distributions as in the classic Shannon framework \cite{Shannon_SecrecySystems_1949}.

\subsection{Security Limits}
The encoding performed by CS is a linear mapping, and as such it cannot completely hide the information contained in a plaintext $x$. This has two main consequences: firstly, linearity propagates scaling. Hence, it is simple to distinguish a plaintext $x'$ from another $x''$ if one knows that $x''=\alpha x'$ for some scalar $\alpha$. For the particular choice $\alpha = 0$ this leads to a known argument \cite[Lemma 1]{Allerton_2008} against the fundamental requirement for perfect secrecy that the conditional PDF $f_{Y|X} (y|x) = f_{Y}(y)$ (\emph{e.g.} in model \EMone). In the following, we will prove that a scaling factor is actually all that can be inferred from the statistical analysis of CS-encoded ciphertexts.

Secondly, linearity implies continuity. Hence, whenever $x'$ and $x''$ are close to each other for a fixed $A$, the corresponding $y'$ and $y''$ will also be close to each other. This fact goes against the analog version of the \emph{diffusion} (or \emph{avalanche effect}) requirement for digital-to-digital ciphers \cite{washington2002introduction}. If the encoding process did not entail a dimensionality reduction, this fact could be exploited every time a plaintext-ciphertext pair $x', y'$ is known. If a new ciphertext $y''$ is available that is close to $y'$, then it is immediately known that the corresponding plaintext $x''$ must be close to $x'$ thus yielding a good starting point for \emph{e.g.} a brute-force attack. 

The fact that $m < n$ slightly complicates this setting since the counterimages of $y''$ through $A$ belong to a whole subspace in which points arbitrarily far from $x'$ exist in principle. Yet, encoding matrices $A$ are chosen by design so that the probability of their null space aligning with $x'$ and $x''$ (that are $k$-sparse w.r.t. a certain $D$) is overwhelmingly small \cite{candes2008restricted}.
Hence, even if with some relaxation from the quantitative point of view, neighbouring ciphertexts strongly hint at neighbouring plaintexts. As an objection to this seemingly unavoidable issue note that the previous argument only holds when the encoding matrix remains the same for both plaintexts, while by our assumption (Section \ref{cryptopersp}) on the very large period of the generated sequence of encoding matrices two neighbouring plaintexts $x', x''$ will most likely be mapped by different encoding matrices to non-neighbouring ciphertexts $y', y''$.

\subsection{Achievable Security Properties}
\subsubsection{Asymptotic Security}
While perfect secrecy is unachievable, we may introduce the notion of \emph{asymptotic spherical secrecy} and show that CS with universal random encoding matrices has this property, \emph{i.e.} no information can be inferred on a plaintext $x$ in model \EMtwo\, from the statistical properties of all its possible ciphertexts but its power. The implication of this property is the basic guarantee that a malicious eavesdropper intercepting the measurement vector will not be able to extract any information on the plaintext except for its power.

\begin{defn}[Asymptotic spherical secrecy]
Let $\Xone$ be a RP whose plaintexts have finite power $0 < W_x < \infty$, $\Yone$ be a RP modelling the corresponding ciphertexts. A cryptosystem has asymptotic spherical secrecy if for any of its plaintexts $x$ 
 and ciphertexts $y$ we have
\begin{equation}
f_{\Yone|\Xone}\left(y|x\right) \mathop{\longrightarrow}_{\mathcal{D}} f_{\Yone|W_x}(y)\label{eq:ASC}\end{equation}
where $\displaystyle\mathop{\rightarrow}_{\mathcal{D}}$ denotes convergence in distribution as $m, n\rightarrow\infty$, $f_{\Yone|W_x}$ denotes conditioning over plaintexts $x$ with the same power $W_x$.
\end{defn}

From an eavesdropper's point of view, asymptotic spherical secrecy means that given any ciphertext $y$ we have \[f_{\Xone|\Yone}\left(x|y\right) \simeq \frac{f_{\Yone|W_x}(y)}{f_{\Yone}(y)} f_{\Xone}(x)\]
implying that any two different plaintexts with an identical, prior and equal power $W_x$ remain approximately indistinguishable from their ciphertexts. In the asymptotic setting, the following proposition holds.

\begin{props}[Asymptotic spherical secrecy of random measurements]
\label{cltasymp}
Let $\Xone$ be a RP with bounded-value plaintexts of finite power $W_x$, $\Yino_j$ any variable of the RP $\Yone$ as in \EMtwo. For $n \rightarrow \infty$ we have
\begin{equation}
f_{\Yino_j|\Xone}(y_j) \mathop{\longrightarrow}_{\mathcal{D}} \mathcal{N}\left(0, W_x\right)
\label{eq:clt}
\end{equation} 
Thus, universal encoding matrices provide independent, asymptotically spherical-secret measurements as in \eqref{eq:ASC}.
\end{props}

The proof of this statement is given in Appendix \ref{acsproofs}. Since the rows of $A$ are independent, the measurements conditioned only on $W_x$ are also independent and Proposition \ref{cltasymp} asserts that, although not secure in the Shannon sense, CS with suitable encoding matrices is able to conceal the plaintext up to the point of guaranteeing its security for $n\rightarrow\infty$. 

As a more empirical illustration of spherical secrecy for finite $n$, we consider an attack aiming at distinguishing two orthogonal plaintexts $x'$ and $x''$ from their encryption (clearly, finite energy must be assumed as in \EMone). The attacker has access to a large number $\chi$ of ciphertexts collected in a set $\mathcal{Y}'$ obtained by applying different, randomly chosen encoding matrices to a certain $x'$ as in \eqref{measures}.
Then, the attacker collects another set $\mathcal{Y}''$ of $\chi$ ciphertexts, all of them corresponding either to $x'$ or to $x''$, and must tell which is the true plaintext between the two. 
\begin{figure}[t]	
\centering
	\subfloat[\label{fig:KS1}$e_{x'}=e_{x{''}}=1$; uniformity test $p$-value $=0.4775$ implies uniformity at $5\%$ significance.]{\includegraphics[clip=true, trim = 5 0 0 5, width = 1.6in]{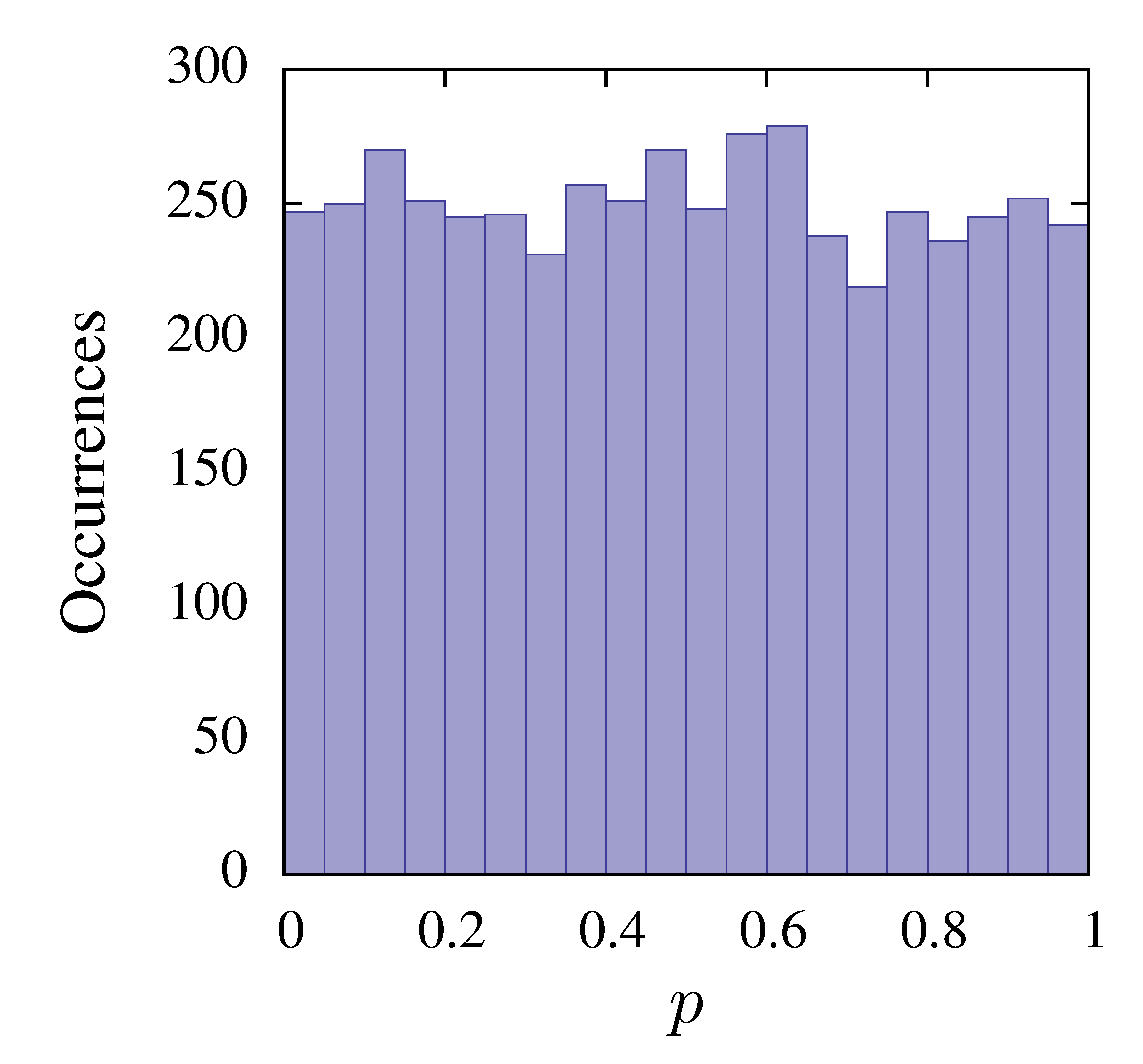}}
	\hspace{1mm}
	\subfloat[\label{fig:KS2}$e_{x'}=1, \, e_{x''}=1.01$; uniformity test $p$-value $\simeq 0$ implies non-uniformity.]{\includegraphics[clip=true, trim = 5 0 0 5, width = 1.6in]{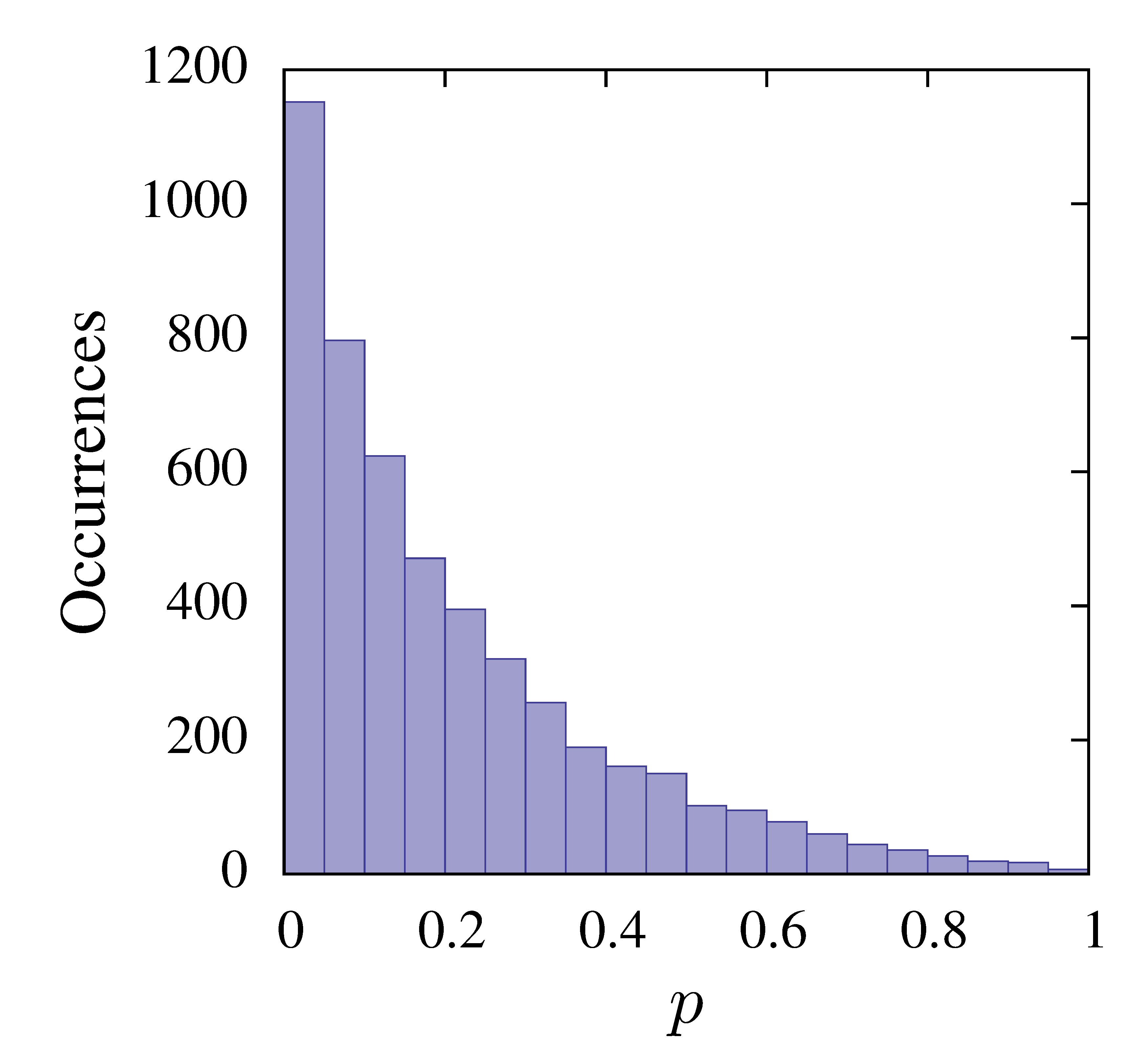}}	
	\caption{Outcome of second-level statistical tests to distinguish between two orthogonal plaintexts $x', x''$. In (a) $x', x''$ have $e_{x'} = e_{x''}$, spherical secrecy applies and the uniform distribution of $p$-values shows that the corresponding ciphertexts are statistically indistinguishable. In (b) $x', x''$ have $e_{x'} \neq e_{x''}$, spherical secrecy does not apply and the distribution of $p$-values shows that the corresponding ciphertexts are distinguishable.}
	\label{fig:secondlevelKS}
\end{figure}
This reduces the attack to an application of statistical hypothesis testing, the null assumption being that the distribution underlying the samples in $\mathcal{Y}''$ is the same as that underlying the samples in $\mathcal{Y}'$. For maximum reliability we adopt a two-level testing approach: we repeat the above experiment for many instances of random orthogonal plaintexts $x'$ and $x''$, performing a two-way Kolmogorov-Smirnov (KS) test to compare the empirical distributions obtained from $\mathcal{Y}'$ and $\mathcal{Y}''$ produced by such orthogonal plaintexts. 

Each of the above tests yields a $p$-value quantifying the probability that two data sets coming from the same distribution exhibit larger differences w.r.t. those at hand.
Given their meaning, individual $p$-values could be compared against a desired significance level to give a first assessment whether the null hypothesis (\emph{i.e.} equality in distribution) can be rejected. Yet, since it is known that $p$-values of independent tests on distributions for which the null assumption is true must be uniformly distributed in $[0,1]$ we collect $P$ of them and feed this second-level set of samples into a one-way KS test to assess uniformity at the standard significance level $5\%$.

This testing procedure is done for $n=256$ in the cases $e_{x'} = e_{x''}=1$ (same energy plaintexts) and $e_{x'}=1, e_{x''}=1.01$, \emph{i.e.} with a $1\%$ difference in energy between the two plaintexts. The resulting $p$-values for $P=5000$ are computed by matching pairs of sets containing $\chi=5 \cdot 10^5$ ciphertexts, yielding the $p$-value histograms depicted in Figure \ref{fig:secondlevelKS}.
We report the histograms of the $p$-values in the two cases along with the $p$-value of the second-level assessment, \emph{i.e.} the probability that samples from a uniform distribution exhibit a deviation from a flat histogram larger than the observed one. When the two plaintexts have the same energy, all evidence concurs to say that the ciphertext distributions are statistically indistinguishable. In the second case, even a small difference in energy causes statistically detectable deviations and leads to a correct inference of the true plaintext between the two.

\subsubsection{Non-Asymptotic Security}
We have observed how asymptotic spherical secrecy has finite $n$ effects (for additional evidence by numerical computation of the Kullback-Leibler divergence see \cite{cambareri2013twoclass}). From a more formal point of view, we may evaluate the convergence rate of \eqref{eq:clt} for finite $n$ to obtain some further guarantee that an eavesdropper intercepting the measurements will observe samples of an approximately Gaussian RV bearing very little information in addition to the energy of the plaintext. We hereby consider $\Xino$ a RV as in \EMone, for which a plaintext $x$ of energy $e_x$ lies on the sphere $S^{n-1}_{e_x}$ of $\Rn$ (with radius $\sqrt{e_x}$).

The most general convergence rate for sums of i.i.d. random variables is given by the well-known Berry-Esseen Theorem \cite{berry1941accuracy} as $O\left(n^{-\frac{1}{2}}\right)$. In our case we apply a recent, remarkable result of \cite{klartag2012variations} that improves and extends this convergence rate to inner products of i.i.d. RVs (\emph{i.e.} any row of $A$) and vectors (\emph{i.e.} plaintexts $x$) uniformly distributed on $S^{n-1}_{e_x}$.  

\begin{props}[Rate of convergence of random measurements]
\label{ROC}
Let $\Xino, Y$ be RVs as in \EMone\, with $A$ a random matrix of i.i.d. zero mean, unit variance, finite fourth moment entries. For any $\rho\in(0,1)$, there exists a subset $\mathcal{F} \subseteq S^{n-1}_{e_x}$ with probability measure $\sigma^{n-1}(\mathcal{F}) \geq 1-\rho$ such that if $x \in \mathcal{F}$ then all $Y_j$ in $Y$ verify
\begin{equation}
\mathop{\text{sup}}_{\substack{\alpha < \beta}} \left| \int^\beta_{\alpha} f_{Y_j|X}(\nu|x) \dd\nu - \frac{1}{\sqrt{2 \pi}} \int^\beta_{\alpha} e^{-\frac{t^2}{2 e_x}} \dd t\right|  \leq \dfrac{C(\rho)}{n}\label{klartagth}\end{equation}
for $C(\rho)$ a non-increasing function of $\rho$.
\end{props}

Proposition \ref{ROC} with $\rho$ sufficiently small means that it is most likely (actually, with probability exceeding $1-\rho$) to observe an $O(n^{-1})$ convergence between $f_{{Y}_j | X}$ and the limiting distribution $\mathcal{N}(0, e_x)$. The function $C(\rho)$ is loosely bounded in \cite{klartag2012variations}, so to complete this analysis we performed a thorough Monte Carlo evaluation of its possible values. In particular, we have taken $10^4$ instances of a RV $X$ uniformly distributed on $S_1^{n-1}$ for each $n = 2^4,2^5, \ldots, 2^{10}$. The PDF $f_{Y_j|X}(y_j|x)$ is estimated with the following procedure: we generate $5 \cdot 10^7$ rows of an i.i.d. Bernoulli random matrix and perform the linear combination in \eqref{measures}, thus yielding the same number of instances of ${Y}_j$ for each $x$ and $n$. On this large sample set we are able to accurately estimate the previous PDF on 4096 equiprobable intervals, and compare it to the same binning of the normal distribution as in the LHS of \eqref{klartagth} for each $(x, n)$. This method yields sample values for \eqref{klartagth}, allowing an empirical evaluation of the quantity ${C}(\rho)$. In this example, when $\rho \geq 10^{-3}$ Proposition \ref{ROC} holds with $C(\rho) = 1.34 \cdot 10^{-2}$.

Hence, straightforward statistical attacks on a CS-encoded ciphertext may only extract very limited information from the plaintext. Yet, other attacks may rely on a larger amount of information, the next level of threat being known-plaintext attacks \cite{washington2002introduction}. These attacks are based on the availability of some plaintext-ciphertext pairs and aim at the extraction of information on the encoding process that can be reused to decode future ciphertexts. Due to its criticality and theoretical depth, the robustness of multiclass CS w.r.t. this class of attacks will be tackled in a separate contribution.

\begin{figure}[t]	
\centering
	\subfloat[$\rho \in (0, 1)$]{\includegraphics[width = 1.6in]{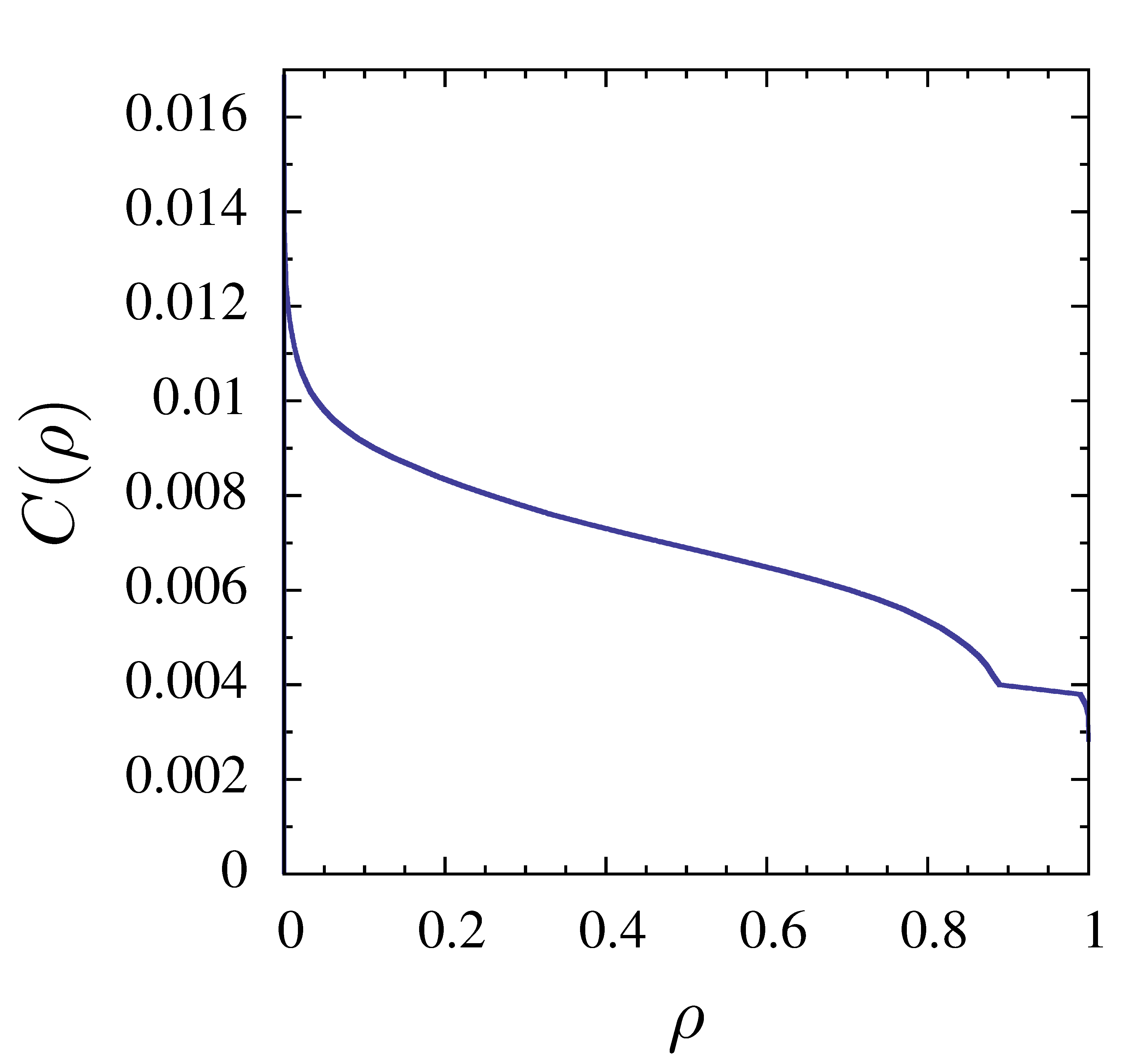}}
	\hspace{1mm}
	\subfloat[$\rho \in (0, 0.01)$]{\includegraphics[width = 1.6in]{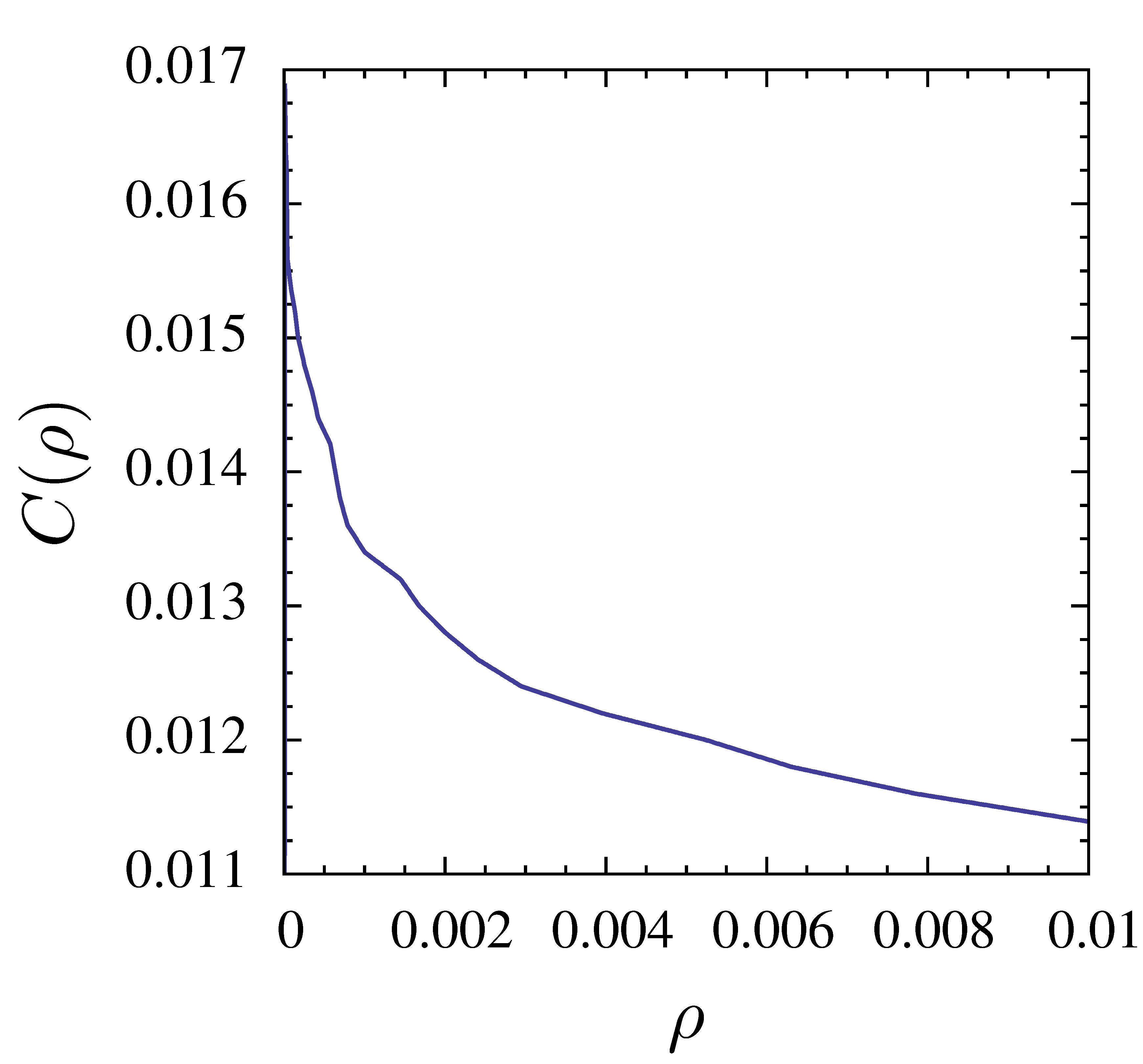}}		
	\caption{Empirical evaluation of $C(\rho)$ in the convergence rate \eqref{klartagth} based on a large number of plaintexts $x$ on the sphere $S^{n-1}_1$ and $n = 2^4,2^5, \ldots, 2^{10}$.}
	\label{fig:crho}
\end{figure}

\section{Application Examples}
\label{section4}
\subsection{Experimental Framework}
In this section we detail some example applications for the multiclass CS scheme we propose. For each example we study the recovery quality attained by first-class receivers against second-class ones in the two-class scheme (Section \ref{twoclssch}). These results encompass the multiclass setting since high-class receivers correspond to first-class recovery performances (\emph{i.e.} $\eta = 0$), while lower-class users attain the performances of a second-class receiver at a fixed $\eta > 0$. 

For each plaintext $x = D s$ being reconstructed the recovery signal-to-noise ratio
$\unit{RSNR}= \frac{\|x\|^2_2}{\|x - \hat{x}\|^2_2}$
with $\hat{x} = D \hat{s}$ denoting the recovered approximation is a common recovery quality index; its average\footnote{$\hat{\mathbb{E}}(\cdot)$ denotes the sample average over a set of realisations of the argument.}, $\unit{ARSNR\, [dB]} = 10 \log_{10} \hat{\mathbb{E}}\left(\textstyle \frac{\|x\|^2_2}{\|x - \hat{x}\|^2_2}\right)$ is then used as an average performance index, and compared against some best- and worst-case curves with the purpose of choosing a suitable perturbation density $\eta$ so that lower-class recovery performances are set to the desired quality level.
 
We complement the previous evidence with an automated assessment of the information content intelligible from $\hat{x}$ by means of feature extraction algorithms. These are equivalent to partially informed attacks attempting to expose the sensitive content inferred from the recovered signal. More specifically, we will try to recover an English sentence from a speech segment, the location of the PQRST peaks in an electrocardiographic (ECG) signal, and printed text in an image. The simulation framework reproducing these tests is available at \url{http://securecs.googlecode.com}. 

\subsubsection{Recovery Algorithms}
While fundamental sparse signal recovery guarantees are well-known from \cite{Candes_2006} when solving BP and BPDN, these convex problems are often replaced in practice by a variety of high-performance algorithms (see \emph{e.g.} \cite{tropp2010computational}). As reference cases for most common algorithmic classes we tested the solution of BPDN as implemented in SPGL$_1$ \cite{van2011sparse,spgl1:2007} against the greedy algorithm CoSaMP \cite{needell2009cosamp} and the generalised approximate message-passing algorithm (GAMP,  \cite{rangan2011generalized}). To optimise these algorithms' performances 
the tests were optimally tuned in a ``genie-aided'' fashion: BPDN was solved as in \eqref{eq:bpdn} with the noise parameter $\gamma = \|\Delta A x\|_2$ as if $(\Delta A, x)$ were known beforehand; CoSaMP was initialised with the exact sparsity level $k$ for each case; GAMP was run with the sparsity-enforcing, i.i.d. Bernoulli-Gaussian prior (see \emph{e.g.}  \cite{vila2011expectation}) broadly applicable in the well-known message passing framework \cite{donoho2009message} and initialised with the exact sparsity ratio $\frac{k}{n}$ of each instance, and the exact mean and variance of each considered test set. 
Moreover, signal-independent parameters were hand-tuned in each case to yield optimal recovery performances. 

For the sake of brevity, in each example we select and report the algorithm that yields the most accurate recovery quality at a lower-class decoder as the amount of perturbation varies. We found that GAMP achieves the highest $\unit{ARSNR}$ in all the settings explored in the examples, consistently with the observations in \cite{vila2011expectation} that assess the robust recovery capabilities of this algorithm under a broadly applicable sparsity-enforcing prior. {Moreover, as $\Delta A$ verifies \cite[Proposition 2.1]{parker2011compressive} the perturbation noise $\varepsilon = \Delta A x$ is approximately Gaussian for large $(m, n)$ and thus GAMP tuned as above yields optimal performances as expected.}
Note that recovery algorithms which attempt to jointly identify $x$ and $\Delta A$ \cite{parker2011compressive, zhu2011sparsity} can be seen as explicit attacks to multiclass encryption and are thus evaluated in a separate contribution, anticipating that their performances are compatible with those of GAMP.

\subsubsection{Average Signal-to-Noise Ratio Bounds}
The perturbation density $\eta$ is the main design parameter for the multiclass encryption scheme, and therefore has to be chosen against a reference lower-class recovery algorithm. To provide criteria for the choice of $\eta$ we adopt two $\unit{ARSNR}$ bounds derived as follows. 

Although rigorous, the lower-class recovery error upper bound of Proposition \ref{ubprops} is only applicable for small values of $(k,\eta)$. To bound typical recovery performances in a larger range we analyse the behaviour of a lower-class decoder that naively (\emph{i.e.} without attempting any attack) recovers $\hat{x}$ such that $y = A^{(0)} \hat{x} = (A^{(0)} + \Delta A) x$, and thus $A^{(0)} (\hat{x}- x) = \Delta A x$. In most cases, such a recovery produces $\hat{x}$ lying close to $x$; we model this by assuming $\|\hat{x}-x\|_2$ is close to be minimum. With this, we may approximate $\hat{x} - x = (A^{(0)})^{+}\Delta A x$, where $\cdot^{+}$ denotes the Moore-Penrose pseudoinverse, that yields $\frac{\|\hat{x}-x\|^2_2}{\|x\|^2_2} \leq \sigma_{\max}((A^{(0)})^{+}\Delta A)^2$. By taking a sample average on both sides, in signal-to-noise ratio our criterion is $\unit{ARSNR} > {\unit{LB}}(m, n, \eta)$ where
\begin{equation}
{\unit{LB}}(m, n, \eta) = \unit[-10 \log_{10} \hat{\mathbb{E}}\left(\sigma_{\max}((A^{(0)})^{+}\Delta A)^2\right)]{dB} \label{eq:practlb}
\end{equation}
${\unit{LB}}(m, n, \eta)$ is calculated in each of the following examples by a thorough Monte Carlo simulation of $\sigma_{\max}((A^{(0)})^{+}\Delta A)$ over $5 \cdot 10^3$ cases. 

The opposite criterion is found by assuming $\unit{ARSNR} < {\unit{UB}}(m, n, \eta)$ where
\begin{equation}
{\unit{UB}}(m, n, \eta) = \unit[-10 \log_{10} \frac{4 \eta m}{(\sqrt{m} + \sqrt{n})^2}]{dB}
\label{eq:practub}
\end{equation}
obtained from a simple rearrangement of \eqref{rerrlb} with $\theta \simeq 1$. We will see how \eqref{eq:practlb} and \eqref{eq:practub} fit the $\unit{ARSNR}$ performances of the examples and provide simple criteria to estimate the range of performances of lower-class receivers from $(m, n, \eta)$.	

\begin{figure}[t]
\centering
	\subfloat[\label{fig:Audio1}]{\includegraphics{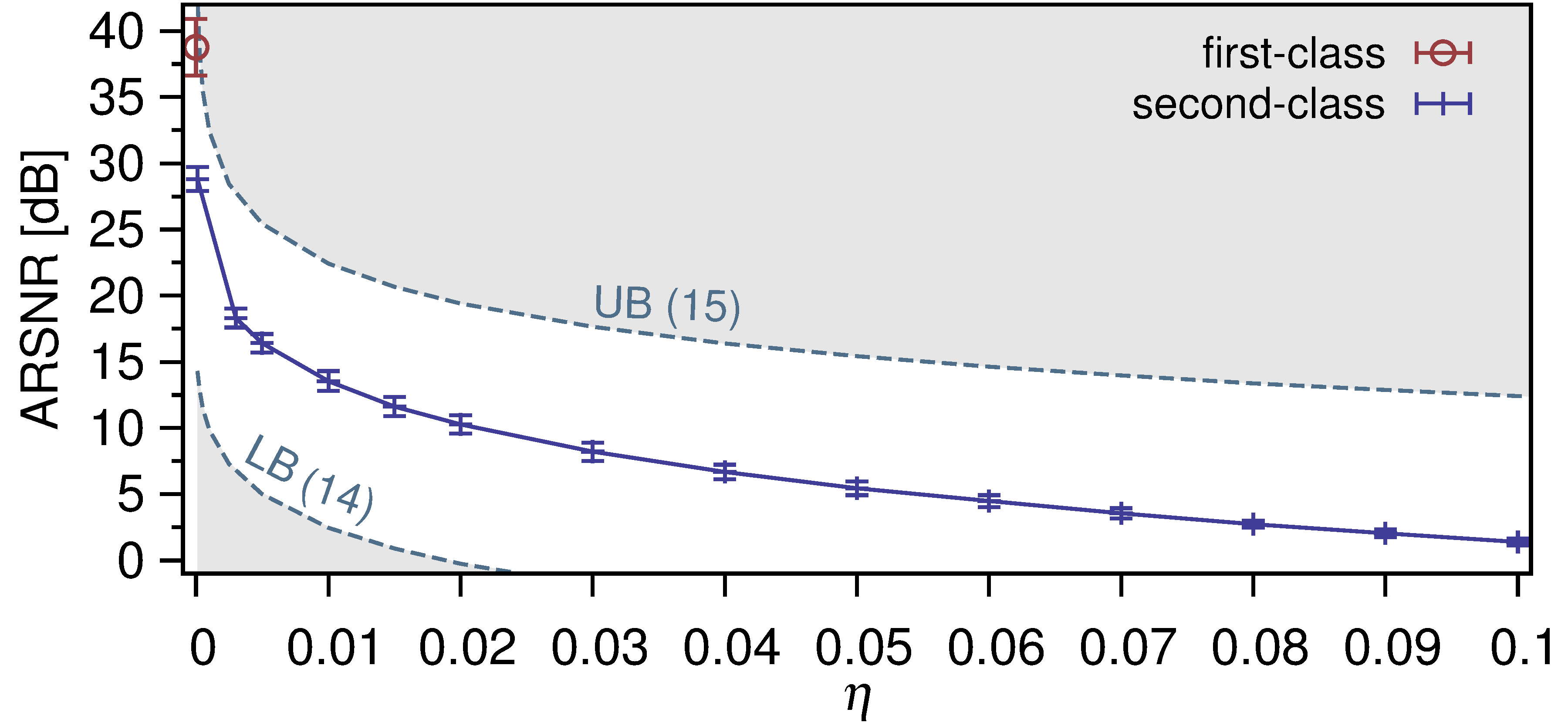}}\\
	\subfloat[\label{fig:Audio2}]{\includegraphics{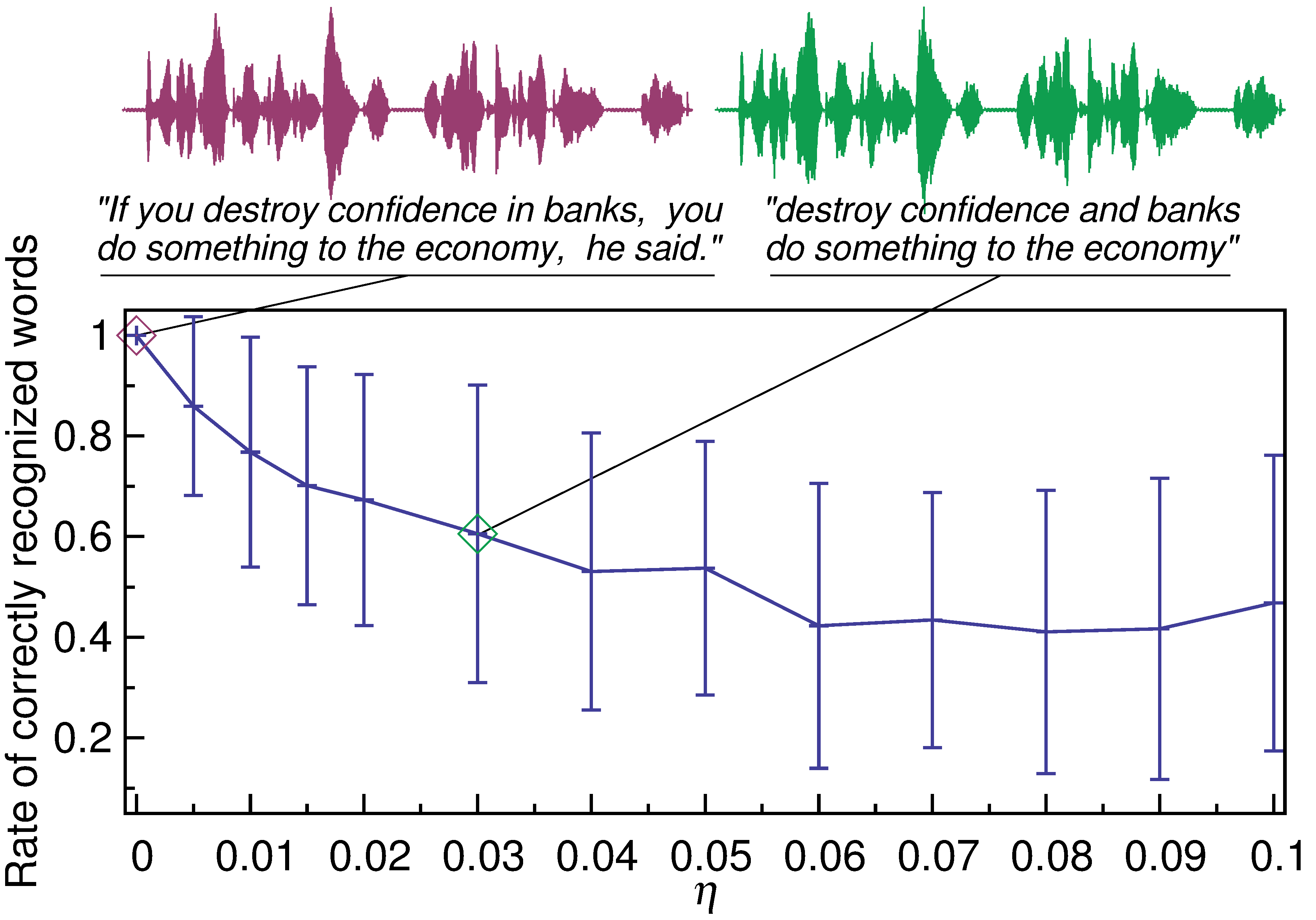}}	
\caption{Multiclass CS of speech signals: (a) Average recovery SNR as a function of the perturbation density $\eta \in [0, 0.1]$ (solid) and second-class $\unit{RSNR}$ upper bound (dashed); (b) Fraction of words exactly recognised by ASR in $\eta \in [0, 0.1]$ (bottom) and typical decoded signals for $\eta = 0, 0.03$ (top).}
\end{figure}

\subsection{Speech Signals}
\label{speechex}

\begin{figure*}[!t]
	\subfloat[\label{fig:Ecg1}]{\centering \includegraphics{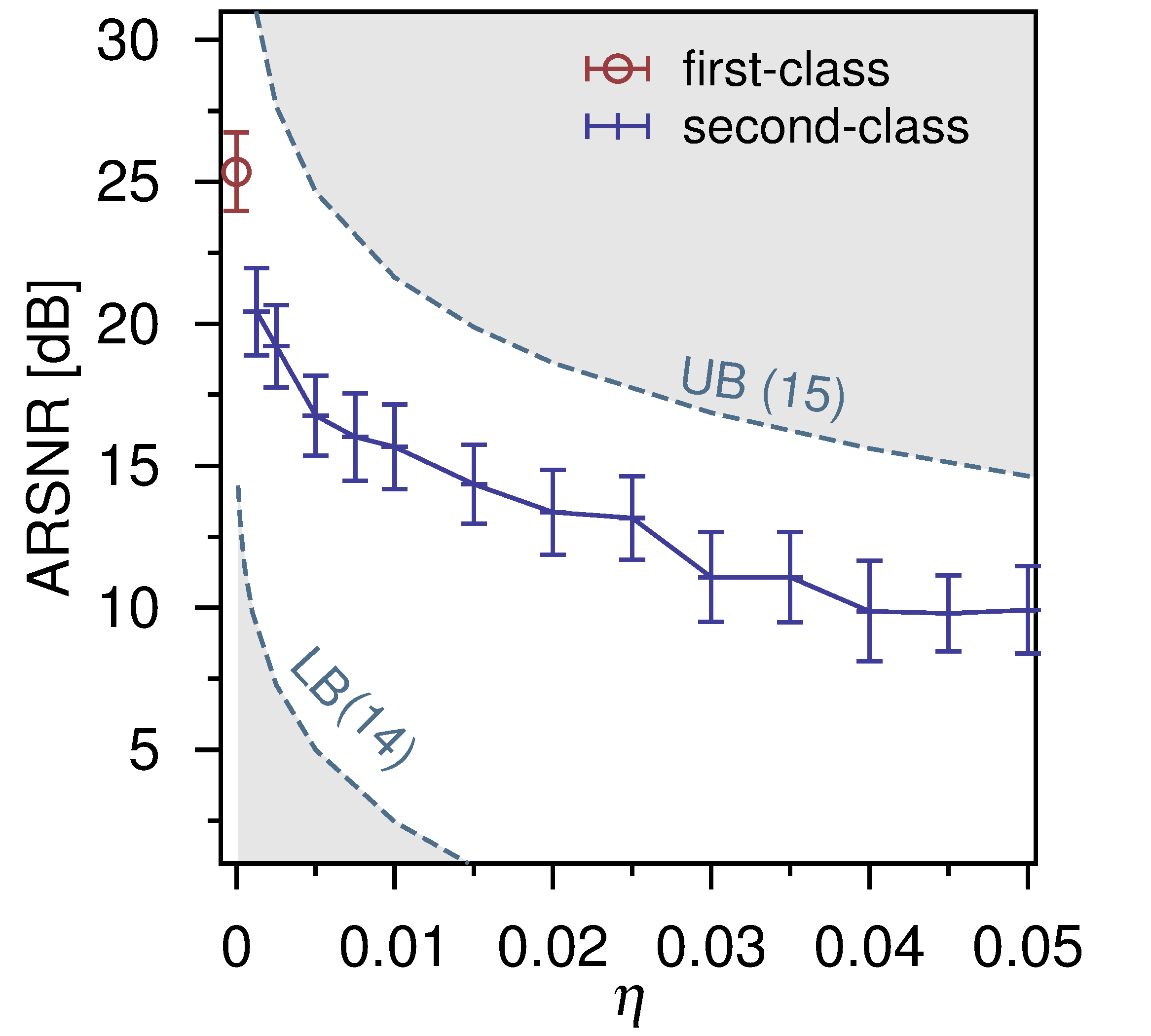}} \hspace{5mm}
	\subfloat[\label{fig:Ecg2}]{\centering \includegraphics{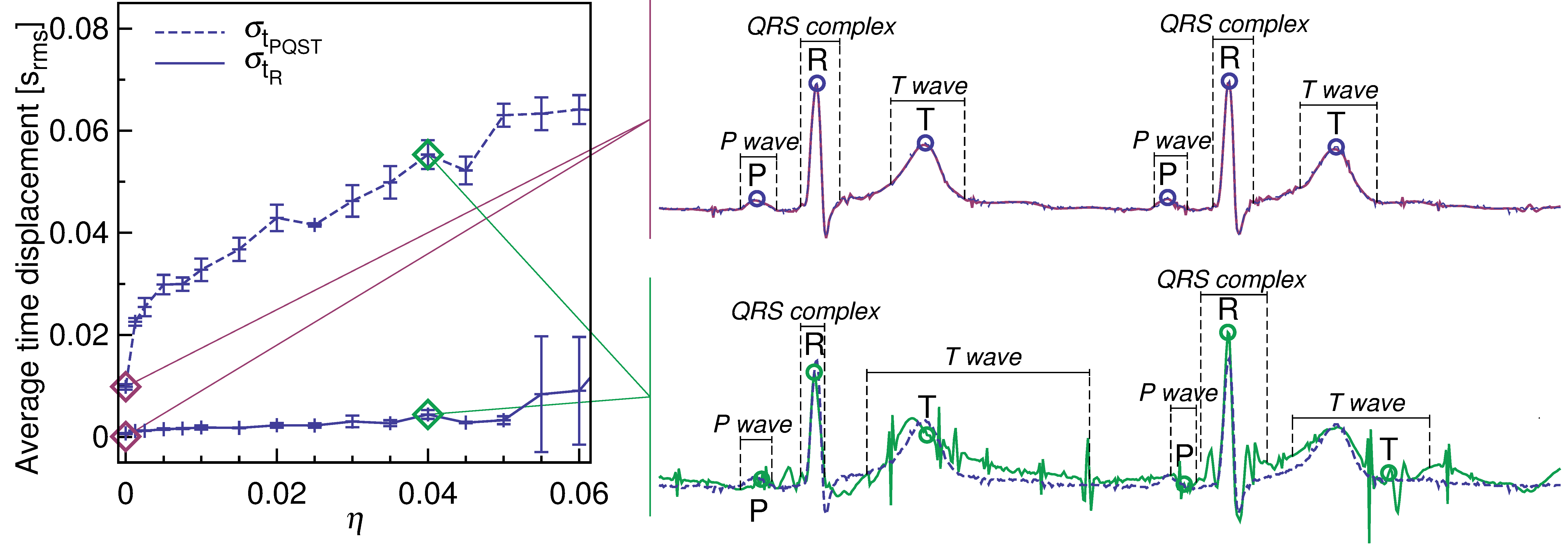}}	
\caption{Multiclass CS of ECG signals: (a) Average recovery SNR as a function of the perturbation density $\eta \in [0, 0.05]$ (solid) and second-class $\unit{RSNR}$ upper bound (dashed); (b) Time displacement (left) of the R (solid) and P,Q,S,T (dashed) peaks as evaluated by APD for $\eta \in [0, 0.05]$ with typical decoded signals (right) for first-class (top) and second-class (bottom) users.}
\end{figure*}

We consider a subset of spoken English sentences from the PTDB-TUG database \cite{PTDB-TUG:2011} with original sampling frequency $f_s=\unit[48]{kHz}$, variable duration and sentence length. 
Each speech signal is divided in segments of $n = 512$ samples and encoded by two-class CS with $m = \frac{n}{2}$ measurements.
We obtain the sparsity basis $D$ by applying principal component analysis \cite{karhunen1947lineare} to $500$ $n$-dimensional segments yielding an ONB. 
The encoding matrix $A^{(1)}$ is generated from an i.i.d. Bernoulli random matrix $A^{(0)}$ by adding to it a sparse random perturbation $\Delta A$ chosen as in \eqref{eq:spperturb} with density $\eta$. The encoding in \eqref{eq:secondclass} is simulated in a realistic setting, where each window $x$ of $n$ samples is acquired with a different instance of $A^{(1)}$ yielding $m$ measurements per speech segment. As for the decoding stage, we apply GAMP as specified above to recover $\hat{x}$ given $A^{(1)}$ (first-class) and $A^{(0)}$ (second-class). 

For a given encoding matrix a first-class receiver is capable of decoding a clean speech signal with $\unit{ARSNR} = \unit[38.76]{dB}$, whereas a second-class receiver is subject to significant $\unit{ARSNR}$ degradation when $\eta$ increases, as shown in Fig.~{\ref{fig:Audio1}}. Note that while the $\unit{RSNR}$ for $\eta = 0$ has a relative deviation of $\unit[2.14]{dB}$ around its mean (the $\unit{ARSNR}$), as $\eta$ increases the observed $\unit{RSNR}$ deviation is less than $\unit[0.72]{dB}$. 
Note how the $\unit{ARSNR}$ values lie in the highlighted range between \eqref{eq:practlb}, \eqref{eq:practub}.

To further quantify the limited quality of attained recoveries, we process the recovered signal with the Google Web Speech API \cite{speechapi, hinton2012deep} which provides basic Automatic Speech Recognition (ASR). The ratio of words correctly inferred by ASR for different values of $\eta$ is reported in Fig.~{\ref{fig:Audio2}}. This figure also reports a typical decoding case: a first-class user (\emph{i.e.} $\eta = 0$) recovers the signal with  $\unit{RSNR} = \unit[36.58]{dB}$, whereas a second-class decoder only achieves a $\unit{RSNR}=\unit[8.42]{dB}$ when $\eta=0.03$. The corresponding ratio of recognised words is $\frac{14}{14}$ against $\frac{8}{14}$. In both cases the sentence is intelligible to a human listener, but the second-class decoder recovers a signal that is sufficiently corrupted to avoid straightforward ASR.

\subsection{Electrocardiographic Signals}

We extend the example in \cite{cambareri2013twoclass} by processing a large subset of ECG signals from the MIT PhysioNet database \cite{PhysioNet} sampled at $f_s=\unit[256]{Hz}$. In particular, we report the case of a typical $25$ minutes ECG track (sequence \verb|e0108|) and encode windows of $n=256$ samples by two-class CS with $m = 90$ measurements, amounting to a dataset of $1500$ ECG instances.
The encoding and decoding stages are identical to those in Section \ref{speechex} and we assume the Symmlet-6 ONB \cite{mallat1999wavelet} as the sparsity basis $D$.

In this setting, the first-class decoder is able to reconstruct the original signal with $\unit{ARSNR}=\unit[25.36]{dB}$, whereas a second-class decoder subject to a perturbation of density $\eta = 0.03$ achieves an $\unit{ARSNR}=\unit[11.08]{dB}$; the recovery degradation depends on $\eta$ as reported in Fig. \ref{fig:Ecg1}.

As an additional quantification of the encryption at second-class decoders we apply PUWave \cite{jane1997evaluation}, an Automatic Peak Detection algorithm (APD), to first- and second-class signal reconstructions. In more detail, PUWave is used to detect the position of the P,Q,R,S and T peaks, \emph{i.e.} the sequence of pulses whose positions and amplitudes summarise the diagnostic properties of an ECG. 
The application of this APD yields the estimated peak instants $\hat{t}_{\rm P,Q,R,S,T}$ for each of $J = 1500$ reconstructed windows and each decoder class, which are afterwards compared to the corresponding peak instants as detected on the original signal prior to encoding. Thus, we define the average time displacement $\sigma_{t} = \sqrt{\frac{1}{J} \sum^{J-1}_{i = 0} (\hat{t}^{(i)} - t^{(i)})^2}$ and evaluate it for $t_{\rm R}$ and $t_{\rm PQST}$. A first-class receiver is subject to a displacement $\sigma_{t_{\rm R}} = \unit[0.6]{ms_{rms}}$ of the R-peak and $\sigma_{t_{\rm PQST}} = \unit[9.8]{ms_{rms}}$ of the remaining peaks w.r.t. the original signal. On the other hand, a second-class user is able to determine the R-peak with $\sigma_{t_{\rm R}} = \unit[4.4]{ms_{rms}}$ while the displacement of the other peaks is $\sigma_{t_{\rm PQST}} = \unit[55.3]{ms_{rms}}$. As $\eta$ varies in $[0, 0.05]$ this displacement increases as depicted in Fig. \ref{fig:Ecg2}, thus confirming that a second-class user will not be able to accurately determine the position and amplitude of the peaks with the exception of the R-peak. 

\subsection{Sensitive Text in Images}
In this final example we consider an image dataset of people holding printed identification text and apply multiclass CS to selectively hide this sensitive content to lower-class users. The $\unit[640 \times 512]{pixel}$ images are encoded by CS in $\unit[10 \times 8]{blocks}$ each of $\unit[64 \times 64]{pixel}$ while the two-class strategy is only applied to a relevant image area of $\unit[3 \times 4]{blocks}$. We adopt as sparsity basis the 2D Daubechies-4 wavelet basis \cite{mallat1999wavelet} and encode each block of $n = \unit[4096]{pixels}$ with $m = 2048$ measurements; the encoding is generated with perturbation density $\eta \in [0, 0.4]$.

The $\unit{ARSNR}$ performances of this example are reported in Fig. \ref{fig:IMG1} as averaged on 20 instances per case, showing a rapid degradation of the $\unit{ARSNR}$ as $\eta$ is increased. This degradation is highlighted in the typical case of Fig. \ref{fig:IMG2} for $\eta = 0.03, 0.2$. 

In order to assess the effect of our encryption method with an automatic information extraction algorithm, we have applied Tesseract \cite{smith2007tesseract}, an optical character recognition (OCR) algorithm, to the images reconstructed by a second-class user. The text portion in the recovered image data is preprocessed to enhance their quality prior to OCR: the images are first rotated, then we apply standard median filtering to reduce the highpass noise components. Finally, contrast adjustment and thresholding yield the two-level image which is processed by Tesseract. To assess the attained OCR quality we have measured the average number of consecutive recognised characters (CRC) from the decoded text image. In Fig. \ref{fig:IMG2} the average CRC is reported as a function of $\eta$: as the perturbation density increases the OCR fails to recognise an increasing number of ordered characters, \emph{i.e.} a second-class user progressively fails to extract text content from the decoded image.

\begin{figure}[!t]
\centering
	\subfloat[\label{fig:IMG1}]{\includegraphics{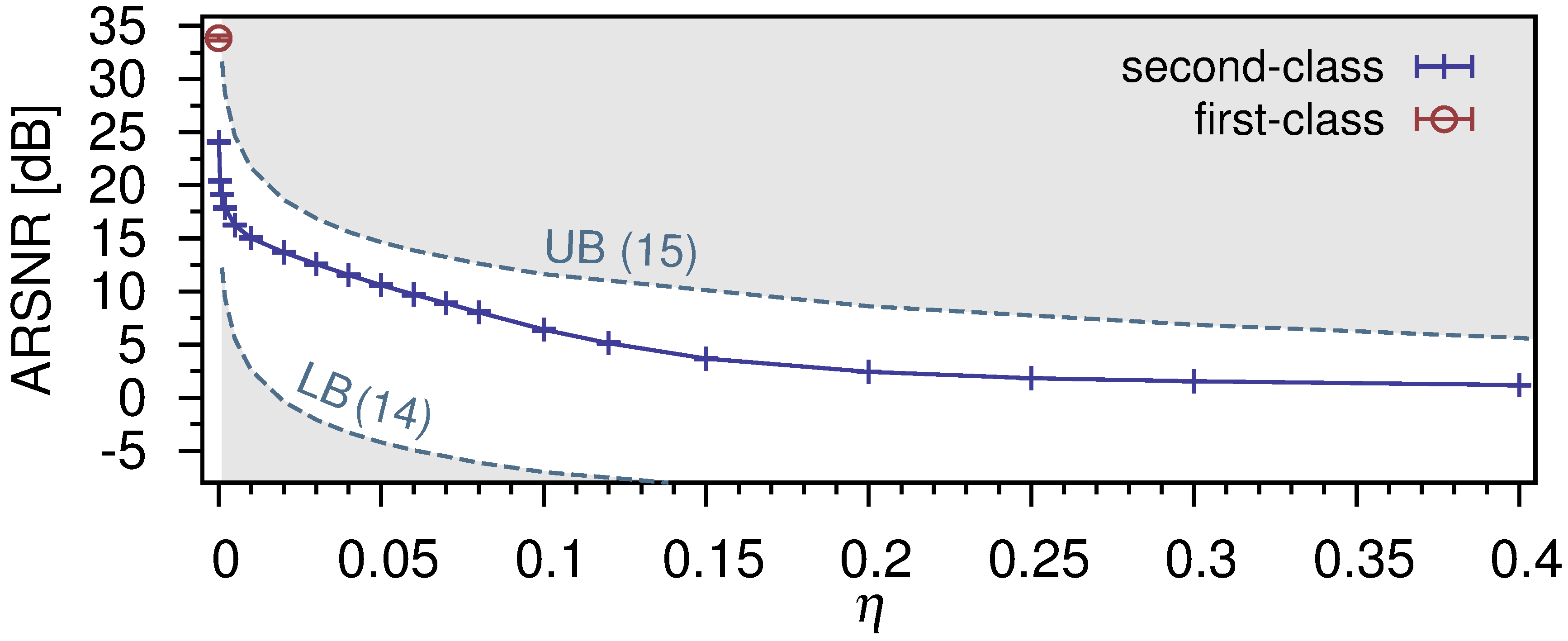}}\\
	\subfloat[\label{fig:IMG2}]{\includegraphics{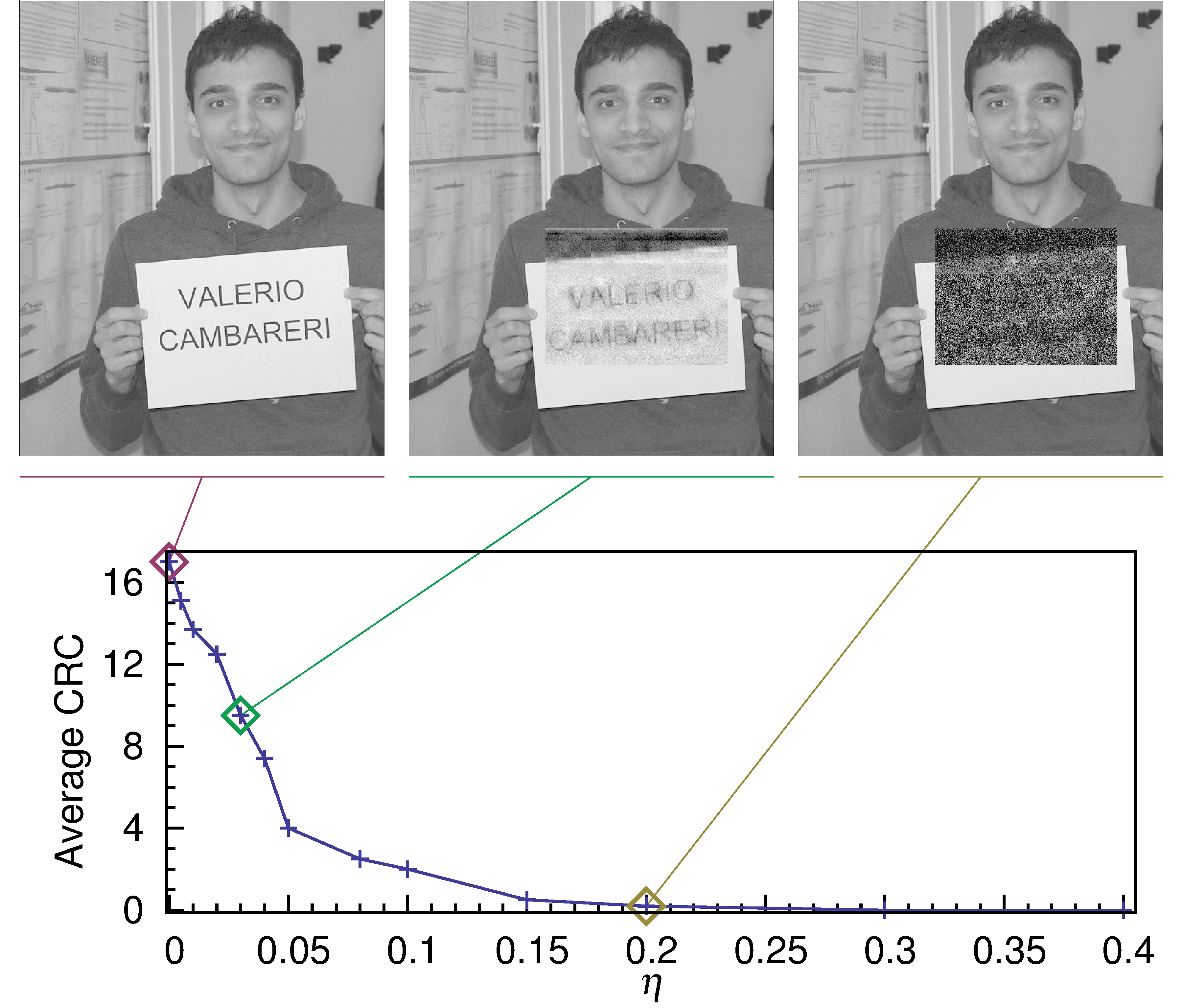}}	
\caption{Multiclass CS of images: (a) Average recovery SNR as a function of the perturbation density $\eta \in [0, 0.4]$ (solid) and second-class $\unit{RSNR}$ upper bound (dashed); (b) Average consecutive recognised characters by OCR for $\eta \in [0, 0.4]$ (bottom) and typical instances for $\eta = 0, 0.03, 0.2$ (top).}
\end{figure}

\section{Conclusion}
Although not perfectly secure, the extremely simple encoding process entailed by CS yields some encryption capabilities with no additional computational complexity, thus providing a limited but zero-cost form of encryption which might be of interest in the design of secure yet resource-limited sensing interfaces. 
In particular, we have shown that when i.i.d. Bernoulli random matrices are used in this linear encoding scheme the plaintext features that leak into the ciphertext (and therefore retrievable by statistical analysis of the latter) are limited to the power of the plaintext as $n \rightarrow \infty$, and thus how an asymptotic definition of secrecy holds for this scheme. In addition, we have given evidence of the $O(\frac{1}{n})$ convergence rate to this limit behaviour. 
We have also detailed how two plaintexts having the same energy and encoded with an i.i.d. Bernoulli random matrix generate statistically indistinguishable ciphertexts, while even small differences in energy are detected by hypothesis testing on the ciphertext for finite $n$. %
The above linear random encoding was modified to envision a multiclass encryption scheme in which all receivers are given the same set of measurements, but are only enabled to reconstruct the original signal with a decoding quality that depends on their class, \emph{i.e.} on the private key they possess. This additional design option amounts to the ability of flipping pseudo-randomly chosen elements of the encoding matrix, and thus represents an appealing alternative to balance the trade-off between the security of the encoded signal and the resources required to provide it.

Finally, the capabilities of multiclass CS were exemplified by simulating the acquisition of sources such as speech segments, electrocardiographic signals and images with the additional security provided by the devised encryption method.

\appendices

\section{Proofs regarding the Second-Class \\ Recovery Error Lower Bound}
We first introduce a Lemma that gives a self-contained probabilistic result on the Euclidean norm of $\varepsilon = \Delta A x$ in \eqref{eq:secondclass}; this is used in the proofs of Theorem \ref{th1} and Corollary 1.
\begin{lem}
\label{lemma2}
Let:
\begin{enumerate} 
\item $\xi$ be a RV with $\mathscr{E}_\xi = \mathbb{E}[\sum^{n-1}_{j=0} \xi^2_j]$, $\mathscr{F}_\xi = \mathbb{E}[(\sum^{n-1}_{j=0} \xi^2_j)^2]$;
\item $\Delta A$ be the sparse random matrix in \eqref{eq:spperturb} with i.i.d. entries and density $\eta = \frac{c}{m n} \leq \frac{1}{2}$.
\end{enumerate}
If $\xi$ and $\Delta A$ are independent, then for any $\theta \in (0, 1)$
\begin{equation}
\mathbb{P}\left(\|\Delta A \xi\|^2_2 \geq 4 m \eta \, \mathscr{E}_\xi \theta \right) \geq \zeta 
\label{eq:lem22}
\end{equation}
with
\begin{equation}
\textstyle \zeta =\left\{1+(1-\theta)^{-2}\left[\left(1+\frac{1}{m}(\frac{3}{2\eta} - 1)\right)\frac{\mathscr{F}_\xi}{\mathscr{E}_\xi^2}-1\right]\right\}^{-1}
\label{eq:lem2}
\end{equation}
\end{lem}
\begin{proof}[Proof of Lemma 1]
Consider $$\|\Delta A \xi \|^2_2 = \sum^{m-1}_{j=0} \sum^{n-1}_{l=0}\sum^{n-1}_{i=0}\Delta A_{j, l}\Delta A_{j, i}\xi_l \xi_i$$ 
We now derive the first and second moments of this positive RV. $\Delta A$ is a random matrix of i.i.d. RVs 
with mean $\mu_{\Delta A_{j,l}} = 0$, variance $\sigma^2_{\Delta A_{j,l}} = 4 \eta$ and $\mathbb{E}[\Delta A^4_{j,l}] = 16 \eta$. 
Using the independence between $\xi$ and $\Delta A$, and the fact that $\Delta A$ is i.i.d. we have
\begin{IEEEeqnarray}{c}
\mathbb{E}\left[\|\Delta A \xi\|^2_2\right] = \sum^{m-1}_{j=0} \sum^{n-1}_{l=0}\sum^{n-1}_{i=0}  \mathbb{E}[\Delta A_{j,l} \Delta A_{j,i}]  \mathbb{E}[\xi_l \xi_i] = \nonumber \\ \sum^{m-1}_{j=0} \sum^{n-1}_{l=0}\sum^{n-1}_{i=0} \sigma^2_{\Delta A} \delta_{l, i} \mathbb{E}[\xi_l \xi_i] = \sum^{m-1}_{j=0} \sigma^2_{\Delta A_{j,l}}  \sum^{n-1}_{l=0}\mathbb{E}[\xi^2_l] = 4 m \eta \, \mathscr{E}_\xi
\nonumber 
\end{IEEEeqnarray}

\noindent For the aforementioned properties of $\Delta A$ we also have
\begin{IEEEeqnarray}{c} { \mathbb{E}[\Delta A_{j,l} \Delta A_{j, i} \Delta A_{v,h} \Delta A_{v, o}]} = \nonumber \\ =
\textstyle{\begin{cases} \sigma^4_{\Delta A}, & \scriptsize{{\begin{cases} j \neq v, l = i, h = o \\ j = v, l = i, h = o, l \neq h \\ j = v,  l = h, i = o, l \neq i \\ j = v,  l = o, i = h, l \neq i \end{cases}}} \\ \mathbb{E}[\Delta A^4_{j,l}], & j = v, l = i = h = o \\ 0,  &\text{otherwise} \end{cases}}\nonumber
\end{IEEEeqnarray}

\noindent that can be used in some cumbersome but straightforward calculations yielding
\begin{equation}\mathbb{E}\left[(\|\Delta A \xi\|^2_2)^2\right] = 16 m \eta (\eta (m-1) \mathscr{F}_\xi + 3 \eta (\mathscr{F}_\xi - \mathscr{G}_\xi) + \mathscr{G}_\xi) \nonumber
\end{equation}
where $\mathscr{G}_\xi = \mathbb{E}\left[\sum^{n-1}_{j = 0} \xi^4_j\right]$. We are now in the position of using a one-sided version of Chebyshev's inequality for positive RVs\footnote{If a r.v. $Z \geq 0$ then $\forall \theta \in (0, 1)$, $\mathbb{P}\left(Z \geq \theta \mathbb{E}[Z]\right) \geq \frac{(1-\theta)^2 \mathbb{E}[Z]^2}{(1-\theta)^2 \mathbb{E}[Z]^2 + \sigma^2_Z}$.}
 to say that, for any $\theta \in (0, 1)$,
\begin{IEEEeqnarray}{c}
\mathbb{P}\left(\|\Delta A \xi\|^2_2 \geq \theta  \mathbb{E}[\|\Delta A \xi\|^2_2] \right) \geq \nonumber \\
\geq \left\{\textstyle 1 + (1-\theta)^{-2}\left[\frac{\mathbb{E}[(\|\Delta A \xi\|^2_2)^2]}{\mathbb{E}[\|\Delta A \xi\|^2_2]^2} - 1 \right]\right\}^{-1} = \nonumber \\
= \left\{\textstyle 1+(1-\theta)^{-2}\left[\left(1-\frac{1}{m}\right) \frac{\mathscr{F}_\xi}{\mathscr{E}_\xi^2} + \frac{3 \eta (\mathscr{F}_\xi - \mathscr{G}_\xi) + \mathscr{G}_\xi}{\eta m \mathscr{E}^2_\xi} - 1 \right]\right\}^{-1} \nonumber
\end{IEEEeqnarray}
which yields \eqref{eq:lem2} by considering that when $\eta \leq \frac{1}{2}$, $3 \eta (\mathscr{F}_\xi - \mathscr{G}_\xi) + \mathscr{G}_\xi \leq \frac{3}{2} \mathscr{F}_\xi$.
\end{proof}

\label{app1}
\begin{proof}[Proof of Theorem \ref{th1}]
Since all decoders receive in absence of other noise sources the same measurements $y = A^{(1)} x$, a second-class decoder would naively assume $y = A^{(0)} \hat{x}$, with $\hat{x}$ an approximation of $x$ obtained by a recovery algorithm that satisfies this equality. Since $A^{(1)} = A^{(0)} + \Delta A$, if we define $\Delta x = \hat{x} - {x}$ we may write
$A^{(0)} x + \Delta A x  = A^{(0)} \hat{x}$
and thus
$A^{(0)} \Delta x = \Delta A x$. $\| \Delta x \|^2_2$ can then be bounded straightforwardly as
$\sigma_{\max}(A^{(0)})^2 \|\Delta x\|^2_2 \geq {\|\Delta A x\|^2_2}$ yielding
\begin{equation}
\|\hat{x} - x\|^2_2 \geq \dfrac{\|\Delta A x\|^2_2}{\sigma_{\max}(A^{(0)})^2}
\label{lb0}
\end{equation}
By applying the probabilistic lower bound of Lemma \ref{lemma2} on $\|\Delta A x \|^2_2$ in \eqref{lb0}, we have that $\|\Delta A x \|^2_2 \geq 4 m \eta \, \mathscr{E}_x \theta$ for $\theta \in (0, 1)$ and a given probability value exceeding $\zeta$ in \eqref{eq:lem2}. Plugging the RHS of this inequality in \eqref{lb0} yields \eqref{finitererrlb}.
\end{proof}

The following Lemma applies to finding the asymptotic result \eqref{rerrlb} of Corollary 1.
\begin{lem}
\label{lemma1}
Let $\Xone$ be an $\alpha$-mixing RP with uniformly bounded fourth moments $\mathbb{E}[X^4_j] \leq m_x$ for some $m_x > 0$. Define ms 
$\mathscr{E}_x = \mathbb{E}\left[\sum^{n-1}_{j = 0} X^2_j\right], \ \mathscr{F}_x = \mathbb{E}\left[\left(\sum^{n-1}_{j = 0} X^2_j\right)^2\right]$.

If $\mathscr{W}_x = \lim_{n \rightarrow \infty} \frac{1}{n} \mathscr{E}_x > 0$ then
$\lim_{n \rightarrow \infty} \dfrac{\mathscr{F}_x}{\mathscr{E}^2_x} = 1$.
\end{lem}
\begin{proof}[Proof of Lemma 2]
Note first that from Jensen's inequality $\mathscr{F}_x \geq \mathscr{E}^2_x$, so $\lim_{n \rightarrow \infty} \frac{1}{n} \mathscr{E}_x > 0$ also implies that $\lim_{n \rightarrow \infty} \frac{1}{n^2} \mathscr{E}^2_x > 0$ and $\lim_{n \rightarrow \infty} \frac{1}{n^2} \mathscr{F}_x > 0$. Since $\lim_{n \rightarrow \infty} \frac{1}{n^2} \mathscr{E}^2_x = \mathscr{W}^2_x > 0$ we may write
\begin{IEEEeqnarray}{c}
\lim_{n \rightarrow \infty} \frac{\mathscr{F}_x}{\mathscr{E}^2_x} = 1 + \textstyle\frac{\lim_{n \rightarrow \infty} {\frac{1}{n^2} \mathscr{F}_x - \frac{1}{n^2} \mathscr{E}^2_x }} {\mathscr{W}^2_x} 
\label{trick}
\end{IEEEeqnarray}
and observe that $\left| \frac{1}{n^2} \mathscr{F}_x -\frac{1}{n^2} \mathscr{E}^2_x \right| \leq \frac{1}{n^2} \sum^{n-1}_{j = 0}\sum_{l=0}^{n-1}  \left|\mathscr{X}_{j,l}\right|$ where
$\mathscr{X}_{j,l} = \mathbb{E}[X^2_j X^2_l] - \mathbb{E}[X^2_j] \mathbb{E}[X^2_l] = \mathbb{E}[(X^2_j - \mathbb{E}[X^2_j])(X^2_l - \mathbb{E}[X^2_l])]$. From the $\alpha$-mixing assumption we know that $\left|\mathscr{X}_{j,l}\right| \leq \alpha(|j-l|) \leq m_x$ and a sequence $\alpha(h)$ vanishing as $h \rightarrow \infty$. Hence,
\begin{IEEEeqnarray}{c}
\left|\frac{1}{n^2} \mathscr{F}_x - \frac{1}{n^2} \mathscr{E}^2_x\right| \leq \frac{1}{n^2} \sum^{n-1}_{j = 0}  \left|\mathscr{X}_{j,j}\right| + \frac{2}{n^2}\sum_{h=1}^{n-1} \sum^{n-h-1}_{j = 0} \left|\mathscr{X}_{j,j+h}\right| \leq \nonumber  \\ \leq  \frac{n \, m_x}{n^2} + \frac{2}{n^2}\sum_{h=1}^{n-1} (n-h) \alpha(h)  \leq  \frac{m_x}{n} + \frac{2}{n}\sum_{h=1}^{n-1} \alpha(h)
\label{expansion}
\end{IEEEeqnarray}
The thesis follows from the fact that the upper bound in \eqref{expansion} vanishes as $n \rightarrow \infty$. This is obvious when $\sum_{h=0}^{+\infty} \alpha(h)$ is convergent. Otherwise, if $\sum_{h=0}^{+\infty} \alpha(h)$ is divergent we may resort to the Stolz-Ces\`aro theorem to find $\lim_{n \rightarrow \infty}\frac{1}{n} \sum_{h=1}^{n-1} \alpha(h) = \lim_{n \rightarrow \infty} \alpha(n) = 0$.
\end{proof} 

\begin{proof}[Proof of Corollary 1]
The inequality \eqref{lb0} in the proof of Theorem 1 is now modified for the asymptotic case of a RP $\Xone$. Note that $A^{(0)}$ is an i.i.d. random matrix with zero mean, unit variance entries; thus, when $m, n \rightarrow \infty$ with $\frac{m}{n} \rightarrow q$ the value $\sqrt{n}\sigma_{\max}(A^{(0)})$ is known from \cite{geman1980limit} since all the singular values belong to the interval
$[1-\sqrt{q}, 1+\sqrt{q}]$. We therefore assume  $\sigma_{\max}(A^{(0)}) \simeq \sqrt{m}+\sqrt{n}$ and take the limit of \eqref{lb0} normalised by $\frac{1}{n}$ for $m, n \rightarrow \infty$, yielding 
\begin{equation}
\lim_{n \rightarrow \infty} \textstyle{\frac{1}{n}{\sum^{n-1}_{j=0}(\hat{x}_j-x_j)^2}}\geq {\displaystyle\lim_{\substack{m, n \rightarrow \infty}}} \frac{\|\Delta A \frac{x^{(n)}}{\sqrt{n}}\|^2_2}{\left(\sqrt{m} + \sqrt{n}\right)^2}
\label{eq:firstlim}
\end{equation}
with $x^{(n)}$ the $n$-th finite-length term in a plaintext $x = \{x^{(n)}\}^{+\infty}_{n = 0}$ of $\Xone$. We may now apply Lemma \ref{lemma2} in $\xi = \frac{x^{(n)}}{\sqrt{n}}$ for each $\|\Delta A \xi\|^2_2$ at the numerator of the RHS of \eqref{eq:firstlim} with $\mathscr{F}_\xi = \frac{1}{n^2}\mathscr{F}_x$, $\mathscr{E}_\xi = \frac{1}{n}\mathscr{E}_x$ and $\mathscr{E}_x, \mathscr{F}_x$ as in Lemma 1. For $m, n \rightarrow \infty$ and $\eta \leq \frac{1}{2}$, the probability in \eqref{eq:lem2} becomes $$\lim_{\substack{m, n \rightarrow \infty}} \zeta = \left\{1+(1-\theta)^{-2}\left[\lim_{n \rightarrow \infty}\frac{\frac{1}{n^2}\mathscr{F}_x}{\frac{1}{n^2}\mathscr{E}^2_x}-1\right]\right\}^{-1}$$
{Since $\Xone$ satisfies by hypothesis the assumptions of Lemma \ref{lemma1}}, then $\lim_{n \rightarrow \infty}\frac{\mathscr{F}_\xi}{\mathscr{E}_\xi^2} = 1$ and $\lim_{\substack{m, n \rightarrow \infty}} \zeta = 1$. Hence, with $\frac{m}{n} \rightarrow q$ and probability $1$ the RHS of \eqref{eq:firstlim} becomes
$$\lim_{m, n \rightarrow \infty} \frac{\|\Delta A \xi\|^2_2}{n (1 + \sqrt{\frac{m}{n}})^2} = \lim_{m, n \rightarrow \infty} \frac{4 \frac{m}{n} \eta \frac{\mathscr{E}_x}{n}}{(1 + \sqrt{\frac{m}{n}})^2} \theta, \ \forall \theta \in (0, 1) $$
and the recovery error power satisfies \eqref{rerrlb}.
\end{proof}

\section{Proofs regarding the Spherical Secrecy \\ of Compressed Sensing}
\label{acsproofs}
\begin{proof}[Proof of Proposition \ref{cltasymp}]
The proof is given by simple verification of the Lindeberg-Feller central limit theorem (see \cite[Theorem 27.4]{billingsley2008probability}) for $Y_j$ in $\Yone$ conditioned on a plaintext $x$ of $\Xone$ in \EMtwo. By the hypotheses, the plaintext $x = \{x_l\}^{n-1}_{l = 0}$ has power $0 < W_x < \infty$ and $x_l^2 \leq M_x$ for some finite $M_x > 0$. Any $\Yino_j|\Xone = \lim_{n \rightarrow \infty} \sum^{n-1}_{l = 0} Z_{j, l}, \, Z_{j, l} = A_{j, l} \frac{x_l}{\sqrt{n}}$ where $Z_{j, l}$ is a sequence of independent, non-identically distributed random variables of moments $\mathbb{E}[Z_{j,l}] = 0, \mathbb{E}[Z^{2}_{j,l}] = \frac{x^{2}_l}{n}$. By letting the partial sum $S^{(n)}_{j} = \sum^{n-1}_{l = 0} Z_{j, l}$, its mean $\mathbb{E}[S^{(n)}_j] = 0$ and $\mathbb{E}[(S^{(n)}_j)^2] = \frac{1}{n} \sum^{n-1}_{l = 0} {x^2_l}$. Thus, we verify the necessary and sufficient condition  \cite[(27.19)]{billingsley2008probability}
$$\lim_{n \rightarrow \infty} \max_{l=0, \ldots, n-1} \frac{\mathbb{E}[Z^2_{j, l}]}{\mathbb{E}[(S^{(n)}_j)^2]} = 0$$
by straightforwardly observing
$$\lim_{n \rightarrow \infty} \max_{l = 0, \ldots, n-1} \frac{\frac{x^2_l}{n}}{\frac{1}{n} \sum^{n-1}_{l = 0} {x^2_l}} \leq \frac{M_x}{W_x} \lim_{n \rightarrow \infty}  \frac{1}{n} = 0$$
The verification of this condition guarantees that $Y_j | \Xone = \lim_{n \rightarrow \infty} S^{(n)}_j$ is normally distributed with variance $\mathbb{E}[(Y_j|\Xone)^2] = \lim_{n \rightarrow \infty} \mathbb{E}[(S^{(n)}_j)^2] = W_x$, \emph{i.e.} $f_{Y_j|\Xone} \displaystyle\mathop{\rightarrow}_{\mathcal{D}} \mathcal{N}(0, W_x)$.
\end{proof} 

\begin{proof}[Proof of Proposition \ref{ROC}]
\label{ROCproof}
We start by considering $Y_j$ in $Y$ of model \EMone\, conditioned on a given $x$ with finite energy $e_x$. Each of such variables is a linear combination \eqref{measures} of $n$ i.i.d. RVs $A_{j,l}$ with zero mean, unit variance and finite fourth moments. The coefficients of this linear combination are $x = \begin{pmatrix} x_0, \cdots,x_{n-1} \end{pmatrix}$ which by now we assume to have $e_x = 1$, \emph{i.e.} to lie on the unit sphere $S^{n-1}_1$ of $\Rn$. Define $\delta = \left(\frac{1}{n} \sum^{n-1}_{l = 0}\mathbb{E}[A^4_{j, l}]\right)^{\frac{1}{4}} < \infty$, which for i.i.d. Bernoulli random matrices is $\delta = 1$, whereas for standard $\mathcal{N}(0, 1)$ random matrices $\delta = 3^{\frac{1}{4}}$. This setting verifies \cite[Theorem 1.1]{klartag2012variations}: for any $\rho \in (0, 1)$ there exists a subset $\mathcal{F} \subseteq S_1^{n-1}$ with measure $\mu(\mathcal{F})$ such that $\frac{\mu(\mathcal{F})}{\mu(S_1^{n-1})} \geq 1-\rho$ and if $x \in \mathcal{F}$, then
\begin{IEEEeqnarray}{c}\mathop{\text{sup}}_{\substack{(\alpha, \beta) \in \mathbb{R}^2 \\ \alpha < \beta}} \left|  \mathbb{P}\left(\alpha \leq \sum^{n-1}_{l=0} A_{j,l} {x}_{l}\leq \beta \right)-\right.\nonumber \\ \left. -\dfrac{1}{\sqrt{2 \pi}}\int^\beta_{\alpha} 
e^{-\frac{t^2}{2}} \dd t\right| \leq \dfrac{C(\rho) \delta^4}{n}\label{klartagproof}\end{IEEEeqnarray}
with $C(\rho)$ a positive, non-increasing function. An application of this result to $x$ with energy $e_x$, \emph{i.e.} on the sphere of radius $\sqrt{e_x}$,  $\delta = 1$ ($A$ i.i.d. Bernoulli) can be done by straightforwardly scaling the standard normal PDF in \eqref{klartagproof} to $\mathcal{N}(0, {e_x})$, thus yielding the statement of Proposition \ref{ROC}.
\end{proof}

\ifCLASSOPTIONcaptionsoff
  \newpage
\fi

\bibliographystyle{IEEEtran}
\bibliography{CompressiveSensingSecure}
\begin{IEEEbiography}
[{\includegraphics[width=1in,height=1.25in,clip,keepaspectratio]{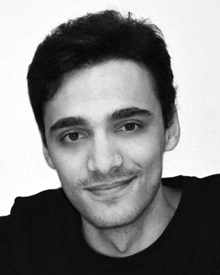}}]{Valerio Cambareri}
(S'13) received the B.S. and M.S. degree (\emph{summa cum laude}) in Electronic Engineering from the University of Bologna, Italy, in 2008 and 2011 respectively. Since 2012 he is a Ph.D. student in Electronics, Telecommunications and Information Technologies at DEI -- University of Bologna, Italy. In 2014 he was a visiting Ph.D. student in the Integrated Imagers team at IMEC, Belgium. His current research activity focuses on statistical and digital signal processing, compressed sensing and computational imaging.
\end{IEEEbiography}%
\vspace{-1.0cm}%
\begin{IEEEbiography}%
[{\includegraphics[width=1in,height=1.25in,clip,keepaspectratio]{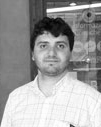}}]
{Mauro Mangia}
(S'09-M'13) received the B.S. and M.S. degree in Electronic Engineering from the University of Bologna, Italy, in 2004 and 2009 respectively; he received the Ph.D. degree in Information Technology from the University of Bologna in 2013. He is currently a post-doc researcher in the statistical signal processing group of ARCES -- University of Bologna, Italy. In 2009 and 2012 he was a visiting Ph.D. student at the \'Ecole Polytechnique F\'ed\'erale de Lausanne (EPFL). 
His research interests are in nonlinear systems, compressed sensing, ultra-wideband systems and system biology. 
He was recipient of the 2013 IEEE CAS Society Guillemin-Cauer Award and the best student paper award at ISCAS2011.
\end{IEEEbiography}%
\vspace{-1.0cm}%
\begin{IEEEbiography}%
[{\includegraphics[width=1in,height=1.25in,clip,keepaspectratio]{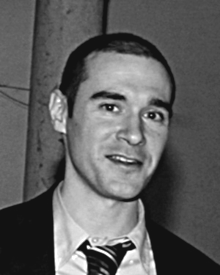}}]
{Fabio Pareschi}
(S'05-M'08) received the Dr. Eng. degree (with honours) in Electronic Engineering from University of Ferrara, Italy, in 2001, and the Ph.D. in Information Technology under the European Doctorate Project (EDITH) from University of Bologna, Italy, in 2007. He is currently an Assistant Professor in the Department of Engineering (ENDIF), University of Ferrara. He is also a faculty member of ARCES -- University of Bologna, Italy. 
He served as Associate Editor for the IEEE Transactions on Circuits and Systems -- Part II (2010-2013). His research activity focuses on analog and mixed-mode electronic circuit design, statistical signal processing, random number generation and testing, and electromagnetic compatibility. He was recipient of the best paper award at ECCTD2005 and the best student paper award at EMCZurich2005.
\end{IEEEbiography}%
\vspace{-1.0cm}%
\begin{IEEEbiography}
[{\includegraphics[width=1in,height=1.25in,clip,keepaspectratio]{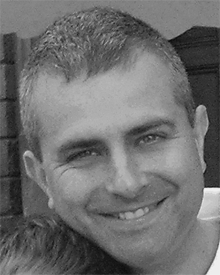}}]
{Riccardo Rovatti}
was born in 1969. He received the M.S. degree in Electronic Engineering and the Ph.D. degree in Electronics, Computer Science, and Telecommunications both from the University of Bologna, Italy in 1992 and 1996, respectively. He is now a Full Professor of Electronics at the University of Bologna. He is the author of approximately 300 technical contributions to international conferences and journals, and of two volumes. His research focuses on mathematical and applicative aspects of statistical signal processing and on the application of statistics to nonlinear dynamical systems. He received the 2004 IEEE CAS Society Darlington Award, the 2013 IEEE CAS Society Guillemin-Cauer Award, as well as the best paper award at ECCTD2005, and the best student paper award at EMCZurich2005 and ISCAS2011. He was elected IEEE Fellow in 2012 {\em for contributions to nonlinear and statistical signal processing applied to electronic systems}.
\end{IEEEbiography}%
\vspace{-1.0cm}%
\begin{IEEEbiography}
[{\includegraphics[width=1in,height=1.25in,clip,keepaspectratio]{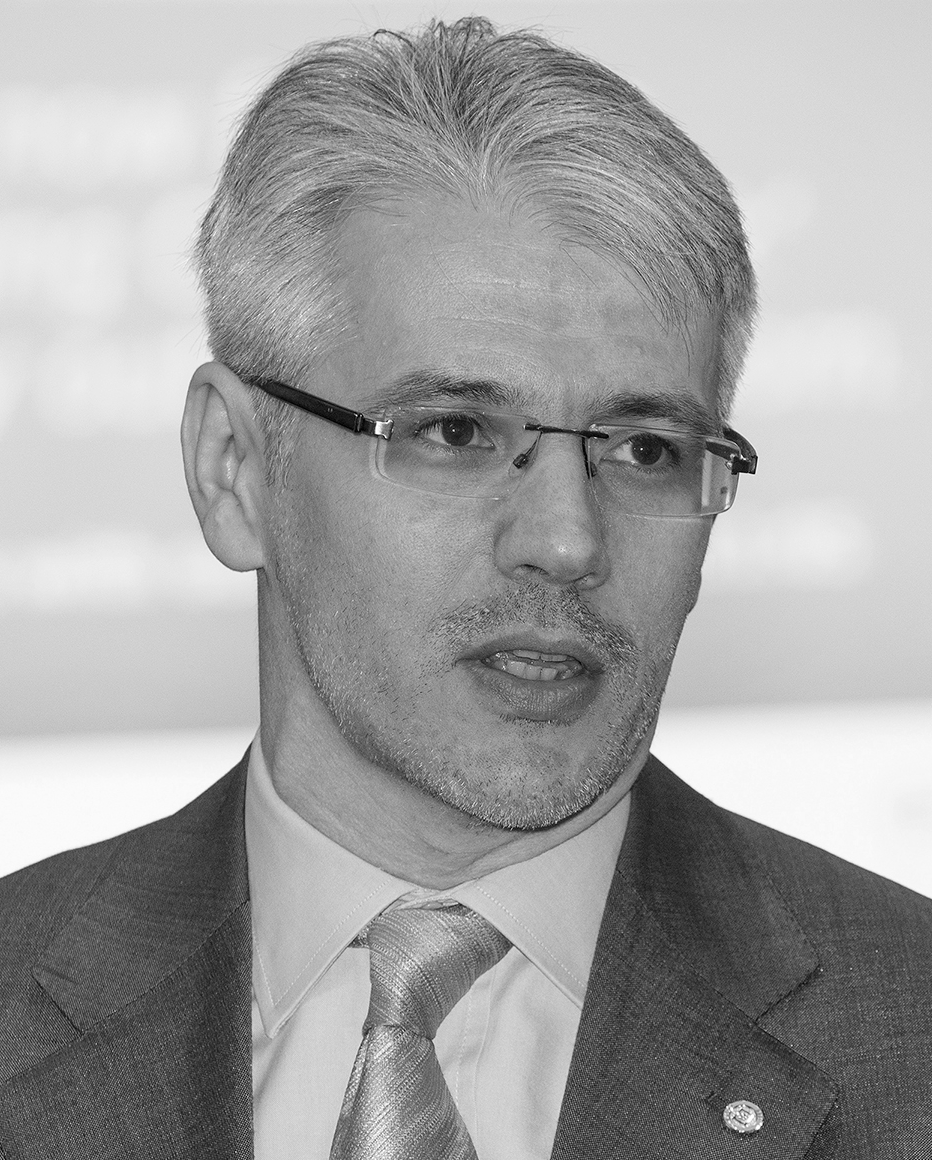}}]
{Gianluca Setti}
(S'89,M'91,SM'02,F'06) received the Ph.D. degree in Electronic Engineering and Computer Science from the University of Bologna in 1997. Since 1997 he has been with the School of Engineering at the University of Ferrara, Italy, where he is currently a Professor of Circuit Theory and Analog Electronics and is also a permanent faculty member of ARCES -- University of Bologna, Italy. 
His research interests include nonlinear circuits, implementation and application of chaotic circuits and systems, electromagnetic compatibility, statistical signal processing and biomedical circuits and systems.
Dr. Setti received  the 2013 IEEE CAS Society Meritorious Service Award and co-recipient of the 2004 IEEE CAS Society Darlington Award, of the 2013 IEEE CAS Society Guillemin-Cauer Award, as well as of the best paper award at ECCTD2005, and the best student paper award at EMCZurich2005 and at ISCAS2011.
He held several editorial positions and served, in particular, as the Editor-in-Chief for the IEEE Transactions on Circuits and Systems -- Part II (2006-2007) and of the IEEE Transactions on Circuits and Systems -- Part I (2008-2009).
Dr. Setti was the Technical Program Co-Chair at ISCAS2007, ISCAS2008, ICECS2012, BioCAS2013 as well as the General Co-Chair of NOLTA2006.
He was Distinguished Lecturer of the IEEE CAS Society (2004-2005), a member of its Board of Governors (2005-2008), and he served as the 2010 President of CASS. He held several other volunteer positions for the IEEE and in 2013-2014 he was the first non North-American Vice President of the IEEE for Publication Services and Products.
\end{IEEEbiography}
\end{document}